%% file: main.tex
\documentclass[11pt]{article}
\usepackage[margin=1in]{geometry}
\usepackage{amsmath,amssymb,amsthm,thm-restate}
\usepackage[hidelinks]{hyperref}
\usepackage[capitalize]{cleveref}
\usepackage{mathtools}
\usepackage{mathrsfs}
\usepackage{microtype}
\usepackage[T1]{fontenc}
\usepackage[utf8]{inputenc} 
\usepackage[switch]{lineno}

\newtheorem{theorem}{Theorem}
\newtheorem{corollary}[theorem]{Corollary}
\newtheorem{conjecture}[theorem]{Conjecture}
\newtheorem{lemma}[theorem]{Lemma}
\newtheorem{claim}[theorem]{Claim}

\theoremstyle{definition}
\newtheorem{definition}[theorem]{Definition}

\usepackage{xcolor}
\usepackage[textsize=tiny]{todonotes}
\newcommand{\CC}{{\mathscr C}}
\newcommand{\DD}{{\mathscr D}}
\newcommand{\BB}{{\mathscr B}}

\newcommand{\PP}{{\mathcal P}}
\newcommand{\N}{\mathbb{N}}

\newcommand{\eps}{\varepsilon}
\newcommand{\Oh}{\mathcal{O}}
\newcommand{\inAppendix}{$\clubsuit$}
\newcommand{\inducedSS}[2]{#1/#2}

\renewcommand{\leq}{\leqslant}
\renewcommand{\le}{\leqslant}
\renewcommand{\geq}{\geqslant}
\renewcommand{\ge}{\geqslant}
\renewcommand{\subset}{\subseteq}
\newcommand{\set}[1]{\{#1\}}
\newcommand{\setof}[2]{\set{#1 \,|\, #2 }}
\newcommand{\from}{\colon}
\newcommand{\res}[1]{{\downharpoonright}_{#1}}
\renewcommand{\cal}{\mathcal}

\DeclareMathOperator{\VCdim}{VCdim}
\DeclareMathOperator{\polylog}{polylog}

\DeclareMathOperator{\subpoly}{subpoly}

\usepackage{enumitem}
\setlist[itemize]{topsep=4pt,itemsep=3pt,parsep=0pt} 
\setlist[enumerate]{topsep=4pt,itemsep=3pt,parsep=0pt}
\setlist[description]{topsep=4pt,itemsep=3pt,parsep=0pt} 

\usepackage{comment}
\usepackage{tcolorbox}
\tcbuselibrary{skins,breakable}

\newif\ifcomments
\commentstrue        

\newlength{\NormalParindent}
\newlength{\NormalParskip}
\AtBeginDocument{%
  \setlength{\NormalParindent}{\parindent}%
  \setlength{\NormalParskip}{\parskip}%
}

\newcommand{\createCommentEnv}[3]{%
  \ifcomments
     \newtcolorbox{#1}{
      breakable, enhanced,
      colback=#3!10,
      colframe=#3,
      boxrule=1pt,
      borderline west={2pt}{0pt}{#3},
      left=10pt,right=10pt,top=5pt,bottom=5pt,
      boxsep=0pt,      
      before skip=\NormalParskip,
      after skip=\NormalParskip,      
      title={\bfseries #2:},
      width=\dimexpr\linewidth+10pt\relax,
      left skip=-10pt,
      before upper={
        \setlength{\parindent}{\NormalParindent}
        \setlength{\parskip}{\NormalParskip}
      },
    }%
  \else
    \excludecomment{#1}%
  \fi
}

\createCommentEnv{szin}{Sz}{red}
\createCommentEnv{janin}{Jan}{green}
\createCommentEnv{RMMin}{RMM}{orange}
\createCommentEnv{claudein}{Claude}{cyan}

\let\phi\varphi
\let\epsilon\varepsilon
\newcommand{\Pp}{\mathcal{P}}

\newcommand{\funding}{
  NM received funding from the European Union through an ERA Fellowship with grant agreement No.\ 101334340 -- LoCoMoDe.
  RM was supported by the National Science Foundation under Grant No.\ DMS-2452111.
  MP was supported by the project BOBR that has received funding from the European Research Council (ERC) under the European Union’s Horizon 2020 research and innovation programme, grant agreement No.\ 948057.
  SzT received funding from the European Research Council (ERC) with grant agreement No.\ 101126229 -- BUKA.\\[0.5em]
\hspace*{\fill}\includegraphics[width=40px]{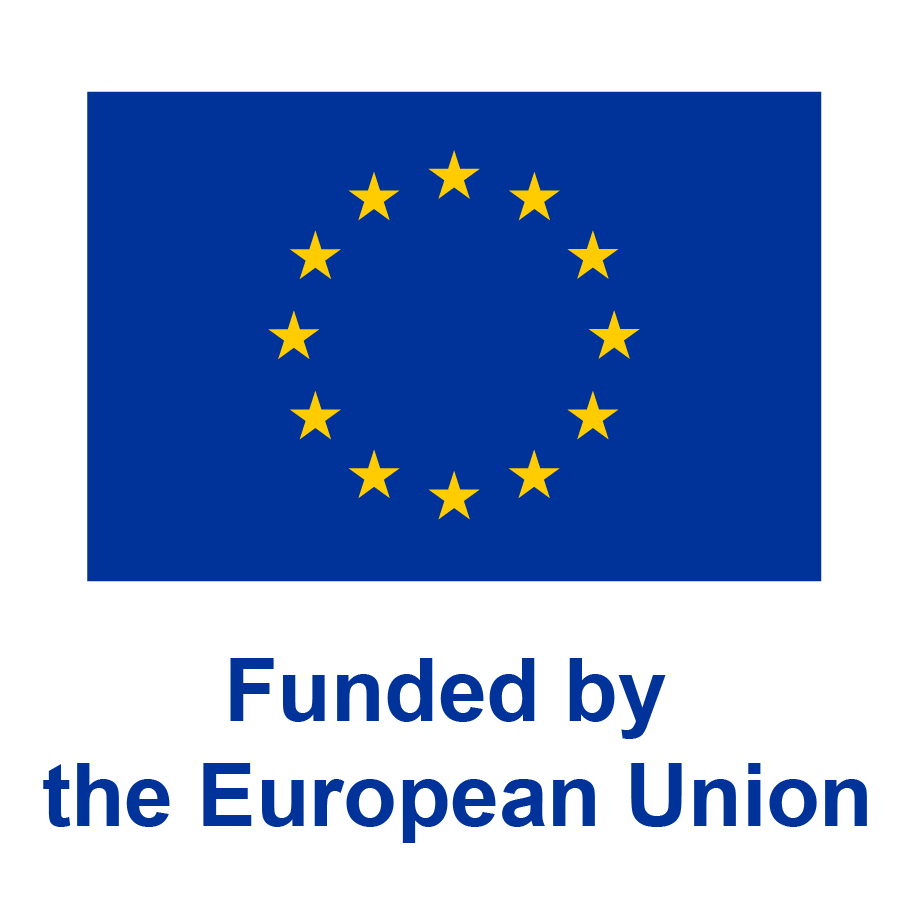}\hspace{1em}\includegraphics[width=40px]{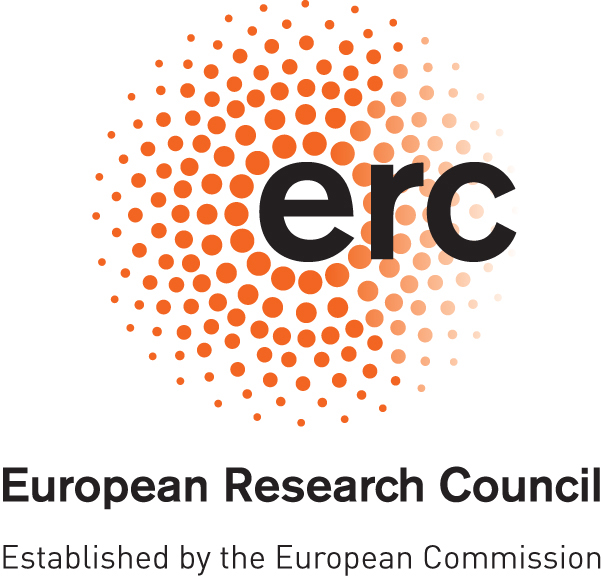}
}

\begin{document}

\title{Neighborhood Complexity and Radius-1 Merge-Width\\ in Monadically Dependent Graph Classes\footnote{\funding}}
 \author{Jan Dreier, Nikolas M{\"a}hlmann, Rose McCarty, Michał Pilipczuk, Szymon Toruńczyk}
\date{\today}
\date{}

\maketitle
\thispagestyle{empty}

\begin{abstract}
Monadic dependence is a proposed structural dividing line for fixed-parameter tractability of first-order model checking on hereditary graph classes. A graph class is \emph{monadically dependent} if the class of all graphs cannot be interpreted in its vertex-colored members using a fixed first-order formula. We prove two structural consequences of monadic dependence. First, every monadically dependent class has \emph{almost linear neighborhood complexity}: for every graph $G$ in the class and every set $A\subseteq V(G)$, the family $\{N_G(v)\cap A : v\in V(G)\}$ has size $|A|^{1+o(1)}$. Second, every $n$-vertex graph in a monadically dependent class has radius-1 merge-width $n^{o(1)}$. Here, merge-width is the decomposition parameter of Dreier and Toru\'nczyk based on construction sequences; its radius-$r$ version measures local reachability among parts through already resolved pairs. This settles the radius-1 case of the conjectured connection between monadic dependence and almost bounded merge-width and 
provides the first decomposition-based structural description of monadically dependent graph classes. Our proof is algorithmic: we give an $\mathcal{O}(n^5)$-time algorithm that, given an $n$-vertex graph $G$ such that $|\{N_G(v)\cap A : v\in V(G)\}|\le \Oh(|A|^d)$ for every $A\subseteq V(G)$, computes a construction sequence witnessing radius-1 merge-width $\mathcal{O}(n^{1-1/d}\log n)$.

\end{abstract}




\setcounter{page}{1}

\section{Introduction}\label{sec:intro}

In the \emph{first-order model checking} problem, the input 
is a graph (or other structure) and a sentence $\varphi$ of first-order logic, and the task is to decide whether \(\varphi\) is true in \(G\).
This fundamental problem captures many concrete problems of interest,
including $k$-Clique, $k$-Dominating Set, and $k$-Independent Set, and
has been the focus of decades of research aimed at understanding which structural restrictions on the input graph render the problem tractable. 

A landmark result of Grohe, Kreutzer, and Siebertz~\cite{grohe2017deciding} established that first-order model checking is fixed-parameter tractable on every \emph{nowhere dense} graph class; the running time of the algorithm is of the form $f(|\varphi|)\cdot n^{1+o(1)}$, where $f$ is a function depending on the class and the $o(1)$ term also depends on the class. For \emph{monotone} (subgraph-closed) graph classes, nowhere denseness is precisely the dividing line between (fixed-parameter) tractability and intractability~\cite{DvorakKT13-journal, grohe2017deciding}. However, there are natural tractable classes --- such as classes of bounded clique-width~\cite{courcelle2000linear} or twin-width~\cite{twwI} --- that are not monotone, contain dense graphs, and hence lie beyond the scope of this classification.

The search for the exact tractability boundary among all \emph{hereditary} (induced-subgraph-closed) graph classes has converged on a notion from Shelah's classification theory: \emph{monadic dependence}~\cite{baldwin1985second}. A graph class $\CC$ is \emph{monadically dependent} if one cannot interpret all graphs in vertex-colored graphs from $\CC$ using a fixed first-order formula. Monadic dependence precisely captures the known tractability boundaries in all settings where a complete classification exists: it is equivalent to nowhere denseness for monotone classes~\cite{adler2014interpreting}, to monadic stability for orderless classes~\cite{nevsetvril2021rankwidth}, and to bounded twin-width for classes of ordered graphs~\cite{twwIV}. This has led to the following conjecture, which is now the central open problem in the area.

\begin{conjecture}[e.g., \cite{warwick-problems,twwIV,ssmc,bddeg}]\label{conj:main}
Let $\CC$ be a hereditary class of graphs. Then the first-order model checking problem is fixed-parameter tractable on $\CC$ if and only if $\CC$ is monadically dependent.
\end{conjecture}

The hardness direction of \Cref{conj:main} was recently confirmed: first-order model checking is AW[$*$]-hard on every hereditary graph class that is not monadically dependent~\cite{flip-breakability}. (Here, AW[$*$] is a parameterized complexity class that can be understood as the parameterized analogue of~PSPACE. In particular, it is strongly believed that AW[$*$]$\neq$FPT.) Tractability has been established for increasingly general subclasses: nowhere dense classes~\cite{grohe2017deciding}, structurally nowhere dense classes~\cite{ssmc}, and monadically stable classes~\cite{stable_MC}.
Yet the full tractability direction for all monadically dependent classes remains open.

Very recently, Dreier and Toru\'nczyk~\cite{mw-mc} suggested a possible line of attack towards \cref{conj:main}. Inspired by \emph{twin-width}~\cite{twwI} and \emph{flip-width}~\cite{flipwidth}, they introduced a family of graph parameters called \emph{merge-width}, indexed by a radius parameter $r\in \N$. They conjectured that every monadically dependent class has \emph{almost bounded merge-width} (i.e., for every fixed $r$, the radius-$r$ merge-width of $n$-vertex graphs in the class is $n^{o(1)}$), and showed that first-order model checking is fixed-parameter tractable on graphs of \emph{bounded merge-width}, provided a suitable decomposition called a \emph{construction sequence} is supplied on input. A possible route towards \cref{conj:main} is then as follows:
\begin{itemize}
\item For a sufficiently high radius $r$ depending on the given sentence $\varphi$, compute a construction sequence of the input graph $G$ whose radius-$r$ merge-width is $|V(G)|^{o(1)}$.
\item Then, use the obtained construction sequence to solve the model checking problem by an extension of the method from~\cite{mw-mc}.
\end{itemize}
In this work, we advance this line of work by the following contributions:
\begin{description}
	\item[\textbf{Result I:}] We prove that monadically dependent graph classes have \emph{almost linear neighborhood complexity}, as defined below.
	\item[\textbf{Result II:}] We leverage the bound on the neighborhood complexity to prove that monadically dependent classes have almost bounded radius-$1$ merge-width. This provides a first step towards the conjecture of Dreier and Toru\'nczyk. Importantly, our proof yields a polynomial-time algorithm that computes a construction sequence of radius-$1$ width $n^{o(1)}$.
\end{description}
We now describe these results and their relevance in more detail.
%

\paragraph{Neighborhood complexity\!\!} is a 
fundamental quantitative notion used in the study of well-structured graph classes. For a graph $G$ and 
vertex set $A\subseteq V(G)$, the \emph{neighborhood complexity} of $A$ is the number of distinct sets of the form $N(v)\cap A$ for $v\in V(G)$, where $N(v)$ denotes the neighborhood of $v$.
We say that a graph class $\CC$ has \emph{almost linear neighborhood complexity} if for every $G\in\CC$ and every $A\subseteq V(G)$,
\[\Big|\set{N(v)\cap A\,:\, v\in V(G)}\Big| \leq |A|^{1+o(1)}.\]
We remark that the $o(1)$ term hides a dependence on the graph class~$\CC$.

In the language of set systems, this says that the set system of neighborhoods in graphs from $\CC$ has an almost linear shatter function, or VC density $1$.
Neighborhood complexity has played a central role in the algorithmic theory of sparse and dense graph classes alike. It is known to be almost linear for nowhere dense classes~\cite{EickmeyerGKKPRS17}, 
and more generally, all monadically stable classes~\cite{stable_MC}, and   all classes of almost bounded merge-width \cite{bonamy2025chiboundednessneighbourhoodcomplexitybounded}.  All these classes are monadically dependent.

Almost linear neighborhood complexity serves as a powerful entry point for establishing a range of quantitative structural properties.
For example, a classic result of Welzl~\cite{Welzl88} shows that every graph class with almost linear neighborhood complexity admits vertex orderings with crossing number \(n^{o(1)}\).
The key observation made in~\cite{stable_MC} was that Welzl orderings with crossing number \(n^{o(1)}\) can be used to efficiently construct \emph{sparse neighborhood covers} with overlap \(n^{o(1)}\). 
These covers serve as a key ingredient in the model checking algorithms for nowhere dense and monadically stable classes~\cite{grohe2017deciding,ssmc,stable_MC}.


Whether monadically dependent classes also have almost linear neighborhood complexity was posed as an open problem in~\cite{stable_MC}.
Our first main result answers this question affirmatively.

\begin{restatable}{theorem}{thmnc}\label{thm:nei-comp-NIP}
Let $\CC$ be a monadically dependent graph class. Then for every $G\in \CC$ and $A\subseteq V(G)$,
\[\Big|\set{N(v)\cap A\,:\, v\in V(G)}\Big| \leq |A|^{1+o(1)}.\]
\end{restatable}
\noindent As usual, the $o(1)$ term hides a dependence on the graph class $\CC$. 
Equivalently, we prove that for every $\epsilon>0$ there is some $c=c(\epsilon, \CC)$ so that the cardinality above is bounded by $c|A|^{1+\epsilon}$.
This result is essentially tight as already nowhere dense classes can have superlinear neighborhood complexity\footnote{The class $\CC \coloneqq \{G \mid \mathrm{maxdegree}(G) \leq \mathrm{girth}(G)\}$ is nowhere dense but has unbounded average degree~\cite[Lem.\ 8]{pilipczuk2017sparsity-lecture1}. The class of $1$-subdivisions of graphs from $\CC$ is still nowhere dense and has superlinear neighborhood complexity.}.

The proof of the corresponding statement for monadically stable classes in~\cite{stable_MC} relies on a reduction to the nowhere dense case using tools from stability theory (Shelah's 2-rank/branching index). 
Our proof of \Cref{thm:nei-comp-NIP} has a similar global structure, but uses induction on the VC-dimension instead of the branching index. This introduces technical difficulties which we overcome using ideas concerning set systems of bounded VC-dimension, originating in the work of Sauer and Shelah \cite{sauer1972sauershelah,shelah1972sauershelah}, and Haussler, Littlestone, Warmuth \cite{Haussler1994,Haussler95} in computational learning theory, such as the unit-distance graph (also called the Hamming graph).

Combined with known results for graphs  of almost linear neighborhood complexity, \Cref{thm:nei-comp-NIP} has several immediate applications.

\begin{corollary}\label{cor:consequences}
  For every monadically dependent graph class $\CC$, every $n$-vertex graph $G\in\CC$ admits:  
  \begin{itemize}
    \item a vertex order (also called a spanning path) with crossing number $n^{o(1)}$ \cite{Welzl88},\\ which can be computed in time $(n+m)\cdot n^{o(1)}$ \cite{dreier2026nearlineartimecomputationwelzl},
    \item a neighborhood cover with overlap $n^{o(1)}$ and strong radius $2$ \cite{stable_MC},
    \item a spanner with stretch $4$ and $n^{1 + o(1)}$ edges\footnote{Take the neighborhood cover from the previous item. For each cluster $C$ with center $v$, keep for each vertex in $C$ the edges of a shortest path to $v$. See for example \cite{peleg1989spanners} for the definition of a spanner.},
    \item an $n^{o(1)}$-bit adjacency labeling scheme/implicit representation \cite{Alon2023,adj-lab},
    \item an algorithm solving All-Pairs Shortest Path in time $n^{2+o(1)}$, \\by combining  
     \cite[Thm. 6.2]{twwIII} and \cite[Lem. 13]{duraj} (see also \cite[Thm. 11]{bonnet2026fastshortestpathgraphs}).
  \end{itemize}
\end{corollary}


Previously, the results mentioned in \cref{cor:consequences} were known in more restrictive regimes where almost linear neighborhood complexity was settled, such as nowhere dense classes~\cite{EickmeyerGKKPRS17}, monadically stable classes~\cite{stable_MC}, and classes of almost bounded merge-width~\cite{bonamy2025chiboundednessneighbourhoodcomplexitybounded}. (The equivalence of the last one with monadic dependence on hereditary classes is the subject of the conjecture of Dreier and Toru\'nczyk~\cite{mw-mc}.)

While neighborhood complexity has rich structural and algorithmic consequences, the overarching goal --- fixed-parameter tractable first-order model checking on monadically dependent classes --- demands more.
Monadic dependence is defined by \emph{forbidding} certain patterns to be interpretable in graphs from the class, a characterization that is inherently non-constructive.
On the other hand, algorithms require \emph{decompositions}. For nowhere dense classes, this role is played by \emph{generalized coloring numbers}, which underlie most algorithmic applications. As discussed above, the working hypothesis is that the right analogue for monadically dependent classes is provided by \emph{merge-width}.

\paragraph{Merge-width\!\!}
    is a family of graph parameters introduced by Dreier and Toru\'nczyk~\cite{mw-mc} that unifies several central structural measures, including treewidth, degeneracy, twin-width, clique-width, and generalized coloring numbers. The underlying notion of a decomposition is called a \emph{construction sequence}: starting from the partition into singletons, it builds the graph by repeatedly either merging two parts, or resolving all currently unresolved vertex pairs between two parts as edges or as non-edges. The \emph{\mbox{radius-$r$} merge-width} measures how many parts any vertex can reach via a path consisting of at most $r$ resolved vertex pairs at any point during the construction.

Graph classes of \emph{bounded merge-width} --- where the radius-$r$ merge-width is bounded by a constant for each fixed $r\in \N$ --- include all classes of bounded expansion and all classes of bounded twin-width, unifying the two frameworks.
The model checking problem for first-order logic is fixed-parameter tractable on graph classes of bounded merge-width, given a witnessing construction sequence~\cite{mw-mc}.
This unifies the model checking results for bounded expansion~\cite{DvorakKT13-journal} and bounded twin-width~\cite{twwI}.

Relaxing this condition, a class $\CC$ has \emph{almost bounded merge-width} if for each fixed $r\in \N$, the $n$-vertex graphs in $\CC$ have radius-$r$ merge-width at most $n^{o(1)}$ (where the ${o(1)}$ term depends on both the radius $r$ and the class $\CC$). This parallels the characterization of nowhere dense classes via almost bounded generalized coloring numbers. It is conjectured~\cite{mw-mc} that almost bounded merge-width coincides with monadic dependence for hereditary classes.
One direction holds: almost bounded merge-width implies monadic dependence~\cite{flip-breakability}. The other direction is open, however our second main result settles it for $r=1$:

\begin{theorem}\label{thm:main}
Let $\CC$ be a monadically dependent graph class.
Then every $n$-vertex graph $G\in \CC$ has radius-$1$ merge-width at most $n^{o(1)}$.
\end{theorem}

Our techniques are algorithmic, not just existential, and they apply to all graphs, not just graphs from monadically dependent classes. Our algorithm returns a construction sequence of small radius-$1$ width for any graph with polynomial neighborhood complexity. More precisely, we prove the following theorem. We note that the value of the constant $k$ in our theorem primarily comes from the constant in Haussler's Packing Lemma~\cite{Haussler95}.

\begin{restatable}{theorem}{mergeWidth}
\label{thm:mergeWidth}
    There is an algorithm that, given an $n$-vertex graph $G$,  computes in time $\mathcal{O}(n^5)$ a construction sequence for $G$ with the following guarantee:
    For all real numbers $c$ and $d\geq 1$, there exists an integer $k = k(c,d)$ so that if $G$ is an $n$-vertex graph such that for every nonempty $A \subseteq V(G)$,
    \[\Big|\set{N(v)\cap A\,:\, v\in V(G)}\Big| \leq c|A|^d,\]
    then the construction sequence has radius-1 merge-width at most $k\cdot n^{1-1/d} \log{n}$.
\end{restatable}


If $G$ is drawn from a monadically dependent class $\CC$,
then \Cref{thm:nei-comp-NIP} guarantees almost linear neighborhood complexity,
and the algorithm computes a construction sequence with \mbox{radius-1} merge-width $n^{o(1)}$, implying \Cref{thm:main}.
Similarly, if $G$ belongs to a class $\CC$ of bounded merge-width, then $G$ has linear neighborhood complexity by~\cite{bonamy2025chiboundednessneighbourhoodcomplexitybounded}, and so the algorithm returns a construction sequence of radius-1 merge-width $\mathcal{O}(\log{n})$ (hiding a constant depending on the class). Notably, the algorithm is oblivious to the values of $d$ and $c$ and its running time is always just~$\mathcal{O}(n^5)$. \Cref{thm:mergeWidth} itself is obtained by a greedy procedure with an iterative reweighing technique similar to the one used by~Welzl~\cite{Welzl88}. 

We did not optimize the running time $\mathcal{O}(n^5)$ of the algorithm. 
A construction sequence 
of radius-1 merge-width $\mathcal{O}(\log{n})$ gives rise to a very particular \emph{signed tree model} with $\mathcal{O}(n\log{n})$ transversal pairs. For graphs of bounded twin-width, signed tree models with $\mathcal{O}(n\log{n})$ transversal pairs can also be constructed in randomized time $\Oh(m+n)\log n$ \cite[Thm. 4]{bonnet2026fastshortestpathgraphs}. 
We leave for future research the question of whether construction sequences of width $\Oh(\log n)$ can be obtained in a similar time for all graphs of linear neighborhood complexity.


\section{Neighborhood Complexity}
\label{sec:neighborhoodComplexity}

In this section we give a high-level overview of the proof of \Cref{thm:nei-comp-NIP}. The complete proof can be found in \Cref{apn:nc}. We begin by introducing the relevant notions. 

\paragraph{Transductions and monadic dependence}
\label{sec:prelims}


A \emph{transduction} $T$ is specified by a number $k$ and a first-order formula $\phi(x,y)$ in the signature consisting of a binary adjacency symbol and $k$ unary relation symbols.
Given a graph $G$, define $T(G)$ as the set of graphs $H$ which can be obtained as follows: first expand $G$ by interpreting the~$k$ unary predicates arbitrarily; then construct a graph on vertex set $V(G)$ and edges $uv$ with $u\neq v$ such that $\phi(u,v)\lor\phi(v,u)$ holds in the expanded structure; finally, take an arbitrary induced subgraph $H$ of the resulting graph.

For example, the two formulas
$\phi_1(x,y) \coloneqq \neg E(x,y)$ and 
$\phi_2(x,y) \coloneqq E(x,y) \vee \exists z\,(E(x,z) \wedge E(z,y))$
specify (with $k=0$) transductions producing the complement and the square of a graph, respectively.

For a graph class $\CC$, we define $T(\CC)\coloneqq \bigcup \setof{T(G)}{G\in\CC}$.
If $\CC$ and $\DD$ are graph classes, then $\CC$ \emph{transduces} $\DD$ if $\DD\subseteq T(\CC)$ for some transduction $T$. A graph class is \emph{monadically dependent} if it does not transduce the class of all graphs. 
(This is the transduction formulation of monadic dependence; the original definition is model-theoretic, and the equivalence follows from the work of Baldwin and Shelah~\cite{baldwin1985second}.)
Since transducibility is transitive (by composing the defining first-order formulas and unary expansions), if $\CC$ is monadically dependent and $\CC$ transduces $\DD$, then also $\DD$ is monadically dependent. 

It is well-known that if $\CC$ is a monadically dependent graph class
that does not contain $K_{t,t}$ as a subgraph, for some $t$,
then $\CC$ is \emph{nowhere dense}.
This follows e.g. by combining the results of 
\cite{adler2014interpreting} and~\cite{dvovrak2018induced}; see \cite[Lemma~35]{flipping-and-forking} or \cite[Lemma~13.7]{maehlmann-thesis} for a direct argument. 
The definition of nowhere denseness will not be relevant here, only the fact that such classes are known to have almost linear neighborhood complexity \cite{EickmeyerGKKPRS17}.
Combining these results, we obtain  the following corollary.

\begin{corollary}
\label{thm:weaklySparseCase}
Fix $t\in\N$ and a monadically dependent class $\CC$ of graphs. Then for every $G \in \CC$ that does not contain $K_{t,t}$ as a subgraph and every $A \subseteq V(G)$,
\[\Big|\set{N(v)\cap A\,:\,v\in V(G)}\Big| \leq 
|A|^{1+o(1)},\]
where the $o(1)$ term depends on both $t$ and $\CC$.
\end{corollary}

\paragraph{Set systems and VC-dimension}
  A \emph{set system} on a set $D$ is a set $\cal F$ of subsets of $D$.  
  If $\cal F$ is a set system on $D$
  and $X\subset D$,
  then by $\cal F\res X$ we denote the set system $\setof{X\cap F}{F\in\cal F}$ on $X$.
  The \emph{VC-dimension} of $\cal F$ is the maximal size 
  of a set $X$ such that $\cal F\res X=2^X$, and is $-\infty$ if $\cal F=\emptyset$.
For a graph $G$, 
we write $\VCdim(G)$ 
to denote the VC-dimension of the set system $\setof{N(v)}{v\in V(G)}$ on $V(G)$.
A graph class $\CC$ has \emph{bounded} VC-dimension if there exists $d \in \N$ such that $\VCdim(G) \leq d$ for all $G\in \CC$.
It is easy to see that every monadically dependent graph class has bounded VC-dimension.



Let $\cal F$ be a set system on $V$, and let $v\in V$.
We say that $A$ and $B$ are a \emph{$v$-pair}
if $A,B\in\cal F$ and $A\triangle B=\set{v}$, that is, if $v$ is the unique element on which $A$ and $B$ differ.
If $A\in\cal F$ is in some $v$-pair then $A$ is \emph{$v$-positive} 
if $v\in A$ and is \emph{$v$-negative} otherwise. The following observation  underlies the inductive proof of the fundamental Sauer-Shelah-Perles lemma (see for instance \cite[Lemma~5.9]{discrepancy}). We leverage this lemma in our proof of \Cref{thm:nei-comp-NIP}. (Claims marked with \inAppendix{} are proved in the~appendix.)

\begin{restatable}[\inAppendix]{lemma}{lemVCind}\label{lem:VC-induction}
  For every nonempty set system $\cal F$ on a set $V$ and every  $v\in V$,
  the following two set systems on~$V$ 
  have VC-dimension strictly smaller than the VC-dimension of $\cal F$:
  \begin{align*}
    \setof{F\in\cal F}{\text{$F$ is $v$-positive}}
    \quad
    \text{and}
    \quad
    \setof{F\in\cal F}{\text{$F$ is $v$-negative}}.
  \end{align*}  
\end{restatable}

\newcommand{\hamming}{\mathrm{Hamming}}
  

The \emph{Hamming graph} of a set system $\cal F$, denoted as $\hamming(\cal F)$, is 
  the graph with vertex set $\cal F$, where two sets 
  $A,B\in\cal F$ are adjacent if and only if $|A\triangle B|=1$.
We use a lemma due to Haussler, Littlestone, and Warmuth \cite[Lemma 2.4]{Haussler1994} (see also \cite[Lemma 2]{Haussler95}). 
\begin{lemma}\label{lem:hamming-deg}
  The Hamming graph of a set system $\cal F$ of VC-dimension $d$ has at most $d|\cal F|$ edges.
\end{lemma}
\begin{corollary}\label{cor:hamming}
  Let $\cal F$ be a set system of VC-dimension $d$, whose Hamming graph $H$ has $m$ edges. Then $H$ has at least $m/d$ non-isolated vertices.
\end{corollary}
\begin{proof}
  Let $\cal G\subset \cal F=V(H)$ be the set of non-isolated vertices of $H$.
  Then $\cal G$ is a set system of VC-dimension at most $d$ with $m=|E(H)|$ Hamming edges,
so $m\le d|\cal G|$ by \Cref{lem:hamming-deg}.
\end{proof}

\paragraph{Bipartite Graphs as Set Systems}
In the proof of \cref{thm:nei-comp-NIP}, it will be enough to consider bipartite graphs $G=(A,B,E)$.
Every bipartite graph $G=(A,B,E)$ \emph{induces} a set system $\setof{N(b)}{b\in B}$ on $A$, and we denote this set system by $\inducedSS{B}{A}$. 
Notice that $|\inducedSS{B}{A}|=|B|$ if $B$ does not contain any pair of \emph{twins}, that is, two vertices with equal neighborhoods.
More generally, for $A'\subset A$ and $B'\subset B$, we view $\inducedSS{B'}{A'}$ as the set system $\setof{N(b)\cap A'}{b\in B'}$ on~$A'$. 




We define two auxiliary bipartite graphs from $G=(A,B,E)$ -- the \emph{positive merge graph} and the \emph{negative merge graph} -- as follows. Both graphs have parts $A$ and $\inducedSS{B}{A}$.
We include edges $\set{a,F}$ with $a\in A$ and  $F\in \inducedSS{B}{A}$ such that $F$ is $a$-positive or $a$-negative, respectively, in the set system $\inducedSS{B}{A}$ on $A$.
For instance, we include $\set{a,F}$ as an edge of the positive merge graph if there exists a set $F' \in \inducedSS{B}{A}$ such that $F \Delta F'=\{a\}$ and $a \in F$. In this case $F'=F-\{a\}$, and we also include $\set{a,F'}$ in the negative merge graph. The intuition is that removing $a$ from the domain causes $F$ and $F'$ to merge, i.e., $F-\{a\}=F'$.
Note that a set $F\in \inducedSS{B}{A}$ is non-isolated in the Hamming graph 
of $\inducedSS{B}{A}$ if and only if it is non-isolated in at least one of the merge graphs of $\inducedSS{B}{A}$.
%
The advantage of merge graphs over the Hamming graph is that they allow us to describe subsets of $B$ using vertices of $A$. Moreover, \cref{lem:VC-induction} yields the following.

\begin{lemma}
\label{lem:blue-vc}
  Let $G=(A,B,E)$ be a bipartite graph and let  $v\in A$.
  Then the neighborhood of $v$ in both the positive merge graph and the negative merge graph 
  is a set system of  strictly smaller VC-dimension than $\inducedSS{B}{A}$.
\end{lemma}



\paragraph{High-level Overview} We now give an overview of the 
proof of \cref{thm:nei-comp-NIP}.
It is enough to show that if $\BB$ is a monadically dependent class of bipartite graphs, and $G=(A,B,E) \in \BB$ has no twins in $B$, then $|B|\le |A|^{1+o(1)}$.
As $\BB$ is monadically dependent, 
the VC-dimension of $\inducedSS{B}{A}$ is bounded by some $d\in\N$  depending only on $\BB$.


The crux of the argument is to find sets $A_1\subset A$ and $B_1\subset B$ with $|B_1|\ge |B|/\polylog(|A|)$, and a partition $\cal P_1$ of $B_1$ so that for each part $P$ of $\cal P_1$,
all vertices of $P$ have distinct neighborhoods on $A_1$, and $\inducedSS{P}{A_1}$ has a VC-dimension at most $d-1$. Moreover, the partition $\cal P_1$ should be definable by a fixed transduction. We now sketch this construction.

We would like the Hamming graph of $\inducedSS{B}{A}$ to have many edges so that we can apply \cref{cor:hamming}. However, even if $|\inducedSS{B}{A}|$ is much larger than $|A|$, its Hamming graph may be sparse or even edgeless
(think of the set system of even-sized subsets of $[n]$). 
The next lemma shows that after restricting to a suitable $A_1\subseteq A$, one obtains a dense Hamming graph.

\begin{lemma}[Special case of Lem. \ref{lem:star}]\label{lem:star0}
  Let $G=(A,B,E)$ be a bipartite graph without twins in $B$ and  $|B|>2$.
  Then there is a set $A_1\subseteq A$ 
  s.t. the Hamming graph of $\inducedSS{B}{A_1}$ has at least $\frac{|B|}{2\ln |A|+2}$~edges.
\end{lemma}

Let $A_1\subset A$ be as in \Cref{lem:star0}
and let $B'\subset B$ be a maximal subset with no twins towards $A_1$,
so that $\inducedSS{B}{A_1}=\inducedSS{B'}{A_1}$.
By \Cref{cor:hamming}, the Hamming graph of $\inducedSS{B'}{A_1}$ has at least 
${|B|}/({d(2\ln |A|+2)})$ non-isolated vertices.
Each such vertex is non-isolated in either the positive or the negative merge graph of $\inducedSS{B'}{A_1}$.
Thus, one of the two merge graphs has at least ${|B|}/({2d(2\ln |A|+2)})$ non-isolated vertices in $\inducedSS{B'}{A_1}$; by symmetry, assume it is the positive merge graph. Note that by \cref{lem:blue-vc}, the neighborhood of each vertex in $A_1$ in the positive merge graph of $\inducedSS{B'}{A_1}$ has VC-dimension smaller than $d$. We now wish to make these neighborhoods disjoint so that we can find a partition of $\inducedSS{B'}{A_1}$. To do so, we use the following lemma from \cite{stable_MC} to find a large subset $B_1\subset B'$ and a set $A_1'\subset A_1$ such that every vertex in $B_1$ has exactly one neighbor in~$A_1'$.

\begin{lemma}[{\cite[Lemma~10]{stable_MC}}]\label{lem:dreier-sampling}
  Let $G=(X,Y,E)$ be a bipartite graph
  with no isolated vertices. Then there are  $X'\subset X$ and $Y'\subset Y$ with $|Y'|\ge \frac{|Y|}{150\ln |X|}$ such that every vertex in $Y'$ has exactly one neighbor in $X'$.
\end{lemma}

By applying \cref{lem:dreier-sampling} to the positive merge graph of
$\inducedSS{B'}{A_1}$ restricted to its non-isolated vertices, we obtain sets $A_1'\subset A_1$ and 
$B_1\subseteq B'$ with
\(
  |B_1|\,\ge\, \Omega\!\left({|B|}/({d\log^2 |A|})\right)
\)
such that every vertex in $B_1$ has exactly one neighbor in $A_1'$ in the
positive merge graph of $\inducedSS{B}{A_1}$. 
Let $\cal P_1$ partition $B_1$ into the neighborhoods (in the positive merge graph) of vertices in $A_1'$.
By
\cref{lem:blue-vc}, for each part $P\in\cal P_1$, the set system $\inducedSS{P}{A_1}$ has VC-dimension at
most $d-1$. The key point is that while we achieved a major progress in reducing the VC-dimension, this came only at a cost of restricting attention to a set $B_1$ that is still large compared with $B$, losing only a polylogarithmic multiplicative factor. So repeating this argument $d$ times preserves a subset of size $|B|/\polylog(|A|)^d$.

We then induct on each part of $\cal P_1$, in parallel. 
On a high level, the resulting process is as follows. Starting from all
vertices of $B$, we keep a large subset $B_i\subseteq B$, a set
$A_i\subseteq A$, and a partition $\cal P_i$ of~$B_i$, so that each part $P\in \cal P_i$ is specified by $i$ vertices of $A$ and has no twins towards $A_i$, and the set system $\inducedSS{P}{A_i}$ has VC-dimension at most
$d-i$. Each step loses only a $\polylog(|A|)$ factor in the size of $B_i$ and lowers the
VC-dimension by one. The next step can be carried as long as $|\cal P_i|<|B_i|/2$, and the process is terminated once $|\cal P_i|\ge |B_i|/2$. This must happen in step $i=d$ at the latest, since 
then for each $P\in \cal P_d$, the set system $\inducedSS{P}{A_d}$ has VC-dimension $0$, so $|P|=1$ and thus $|\cal P_d| = |B_d|$.
Once the process terminates, we have that:
\begin{itemize}
  \item $B_i$ is still almost as large as $B$: $|B_i|\ge |B|/\polylog_d(|A|)$, and
  \item $\cal P_i$ contains many parts: $|\cal P_i|\ge |B_i|/2$.
\end{itemize}

The bipartite graph with sides $\cal P_i$ and $A$, where each 
part $P\in\cal P_i$ is adjacent to the~$i\le d$ vertices encoding it, is $K_{d+1,d+1}$-free and transducible from $G$ by a fixed transduction.
Applying \Cref{thm:weaklySparseCase} then gives 
$|\cal P_i|\le |A|^{1+o(1)}$. Combining the inequalities yields
$|B|\le |A|^{1+o(1)}$, as desired.

\section{Radius-1 Merge-Width}
\label{sec:merge-width}

This section is dedicated to proving \Cref{thm:mergeWidth}, which shows how to efficiently compute a construction sequence with small radius-1 merge-width. We begin by defining the relevant terms.

\paragraph{Merge-Width}

Consider a vertex set $V$.
A \emph{construction sequence} is a sequence of steps, maintaining a partition $\cal P$~of~$V$ and a partition of  ${V\choose 2}$ into three sets:  
\emph{edges} $E$, \mbox{\emph{non-edges}~$N$},
and \emph{unresolved} pairs~$U$.
 Initially, $\cal P$ partitions $V$ into singletons, and every pair in $V\choose 2$ is unresolved.
In each step, one of three operations is performed:
\begin{itemize}
  \item  \emph{merge} two parts $A,B\in\cal P$, \(A \neq B\), replacing the two parts by their union $A\cup B$, 
  \item \emph{resolve positively} a pair of parts $A,B\in\cal P$ (possibly $A=B$), declaring 
 all the unresolved pairs $\set{a,b}\in U$ with $a\in A,b\in B$ as \emph{edges}, that is, moving them from \(U\) to \(E\), or 
  \item \emph{resolve negatively} a pair of parts $A,B\in\cal P$ (possibly $A=B$), declaring 
     the corresponding unresolved pairs as \emph{non-edges}, that is, moving them from \(U\) to \(N\).
\end{itemize}
In the end, we require that $\cal P$ has one part, and that every pair from ${V\choose 2}$ is resolved as either an edge or a non-edge.
We thus say this is a construction sequence of the graph $G=(V,E)$, where \(E\) is the final edge set of the process. 

The \emph{radius-$r$ width}
of a construction sequence is the smallest integer $k$ such that at every step in the sequence, for every vertex $v\in V$, at most $k$ parts
of the current partition $\cal P$ can be reached from $v$ by a path of length ${\le}\,r$ in the graph $(V,E\cup N)$.
Finally, the \emph{radius-$r$ merge-width} of a graph~$G$
is the minimum radius-$r$ width over all construction sequences of $G$. 

Observe that for any construction sequence of a graph $G$, at any given step, for any two parts $A,B$ of the partition $\cal P$, the unresolved vertex pairs $ab$ with $a\in A$ and $b\in B$ are either all edges of~$G$, or are all non-edges in $G$.
See \Cref{fig:b} for an example illustrating radius-$1$ merge-width.

{\definecolor{lightred}{RGB}{255, 158, 159}
\definecolor{darkred}{RGB}{176, 0, 2} 
\definecolor{lightergray}{RGB}{234,234,234}
\definecolor{darkergray}{RGB}{126,126,126}
\begin{figure}
    \centering
    \includegraphics[scale=0.9]{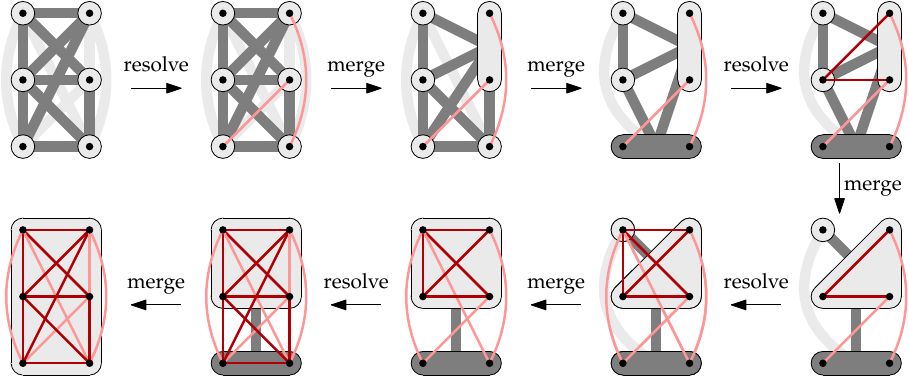}
    \caption{
    A construction sequence of a graph witnessing radius-1 merge-width at most three.
    The \emph{resolved pairs} are drawn as  \textcolor{darkred}{\rule[0.5ex]{0.3cm}{1.3pt}} (edges) and \textcolor{lightred}{\rule[0.5ex]{0.3cm}{1.3pt}} (non-edges).
    The \emph{unresolved pairs} are drawn in an aggregated way as \textcolor{darkergray}{\rule[-0.ex]{0.3cm}{5.3pt}} (edges) 
    and \textcolor{lightergray}{\rule[-0.ex]{0.3cm}{5.3pt}} (non-edges) connecting parts,
    and as
    background colors within each part.
    The first and sixth displayed stages each present several resolve steps.
    Figure replicated with permission from~\cite[Fig.\ 4]{mw-mc}.
}\label{fig:b}
\end{figure}
}
\paragraph{Fractional Twins}
Our algorithm is built on techniques from computational geometry and the combinatorics of set systems related to VC-dimension and $\delta$-separated set systems. In particular, our algorithm uses the \emph{multiplicative weight update} technique of Welzl~\cite{Welzl88, ChazelleWelzl89}. In order to apply this technique, we now describe how graphs with small neighborhood complexity contain weighted near twins, i.e. pairs of vertices with almost the same neighborhood.

For a graph $G$ with distinct vertices $u$ and $v$, we write $\Delta_G(u,v)\coloneqq (N(u)\triangle N(v))\cup\set{u,v}$ for the set of all vertices that are either equal to $u$ or $v$ or adjacent to exactly one of $u$ and $v$. 
We omit the subscript whenever the graph is clear from context.

For a positive real number $p$, we say that a graph $G$ has \emph{$p$-fractional twins} if for every weight function $w:V(G) \rightarrow \mathbb{Z}^+$, there exist distinct vertices $u,v \in V(G)$ such that \begin{align}\label{eq:frac-twins} w\left(\Delta(u,v)\right) \leq p \cdot w(V(G)).\end{align}For a set $S\subset V(G)$ we denote $w(S)\coloneqq \sum_{v \in S}w(v)$. Any  pair of vertices $u$ and $v$ satisfying \eqref{eq:frac-twins}  is called a pair of \emph{$p$-fractional twins with respect to $w$}. 

The following lemma is a corollary of Haussler's Packing Lemma~\cite{Haussler95}.

\begin{restatable}[\inAppendix{}]{lemma}{existFractionalTwins}
\label{lem:existFractionalTwins}
For all real numbers $c$ and $d\geq 1$, there exists an integer $r = r(c,d)$ so that if $G$ is an $n$-vertex graph such that $n\geq 2$ and for every nonempty $A \subseteq V(G)$,
    \[\Big|\set{N(v)\cap A\,:\,v\in V(G)}\Big| \leq c|A|^d,\]
then $G$ has $(r \cdot n^{-1/d})$-fractional twins.
\end{restatable}

Next we show how to find a single pair of fractional twins with respect to two different weight functions simultaneously. We note that the auxiliary function $w$ below is just a weighted average. 


\begin{restatable}[\inAppendix{}]{lemma}{twoWeightFunctions}
\label{lem:twoWeightFunctions}
    Let $p>0$, let $G$ be a graph with $p$-fractional twins, and let $w_1, w_2: V(G) \rightarrow \mathbb{Z}^+$. Then there exist distinct vertices $u$ and $v$ that are $3p$-fractional twins with respect to both $w_1$ and $w_2$. Moreover, given an arbitrary graph $G$ and functions $w_1,w_2: V(G) \rightarrow \mathbb{Z}^+$ as input (i.e., without any knowledge of $p$), such vertices $u$ and $v$ can be found by computing an auxiliary weight function $w$ and selecting $u$ and $v$ to minimize $w(\Delta(u,v))$. If $G$ has at most $n$ vertices and every vertex has weight at most $2^n$ with respect to both $w_1$ and $w_2$, then this algorithm runs in time $\mathcal{O}(n^4)$.
\end{restatable}

\paragraph{Finding Construction Sequences}

Now we prove the main theorem about radius-1 merge-width.

\mergeWidth* \begin{proof} 

    The algorithm directly constructs the construction sequence in $n-1$ rounds, each consisting of several resolve operations followed by a single merge. We write $\Pp_i$ and $R_i\subset {V\choose 2}$ for the partition and the set of resolved vertex pairs at the beginning of round~$i=0,1,\ldots,n-2$, where $|\cal P_i|=n-i$. Along with these, we maintain the following additional information:
    \begin{itemize}
    \item a \emph{leader} $L_i(P) \in P$ for each part $P \in \Pp_i$; we write $L_i(x)$ for the leader of the part containing~$x$;
    \item the \emph{leader graph} $G_i$, the subgraph of $G$ induced by the leaders $\{L_i(P): P \in \Pp_i\}$; and
    \item a weight function $w_i\colon V(G_i) \rightarrow \mathbb{Z}^+$.
    \end{itemize}
    
    Let $r=r(c,d)$ denote the integer from Lemma~\ref{lem:existFractionalTwins}. Since the assumption of the theorem holds not only for $G$ but trivially also for every induced subgraph of $G$, by Lemma~\ref{lem:existFractionalTwins} every induced subgraph of $G$ with $\ell\geq 2$ vertices has $(r \cdot \ell^{-1/d})$-fractional twins. In particular, since at every stage we merge exactly two parts, the leader graph $G_i$ has exactly $n-i$ vertices. Thus $G_i$ has $r(n-i)^{-1/d}$-fractional twins. We now give the algorithm.

    \begin{tcolorbox}[breakable, enhanced, colback=black!5, colframe=black!40, boxrule=0.5pt, left=6pt, right=6pt, top=6pt, bottom=6pt]
    \noindent\textbf{The algorithm.} Let $G_0=G$ and $\mathcal{P}_0$ be the partition of $V(G)$ into $n$ singleton sets, and $R_0=\emptyset$. For each $P \in \mathcal{P}_0$, let its leader $L_0(P)$ be the unique vertex $x \in P$ and set $w_0(x)=1$. 

    \medskip

    \noindent For $i=0,1,2,\ldots, n-2$, do the following.\begin{enumerate}[leftmargin=*]
    \item\label{itm:twoWeightsStep} Apply \Cref{lem:twoWeightFunctions} to the weight function $w_i$ of $G_i$ and the weight function $w_\mathbf{1}$ assigning every vertex of $G_i$ weight 1; that is, compute a weighted average of these two weight functions and then find distinct vertices $u_i$ and $v_i$ so that $\Delta_{G_i}(u_i, v_i)$ has minimum weight with respect to this weighted average. Thus we are guaranteed that\begin{align}\label{equationMUW} w_i(\Delta_{G_i}(u_i, v_i)) &\leq 3r (n-i)^{-1/d} \cdot w_i(V(G_i)),\textit{\quad and}\\
      \label{equationSize}
    |\Delta_{G_i}(u_i, v_i)|
    &=
    w_\mathbf{1}(\Delta_{G_i}(u_i, v_i))
    \leq
    3r (n-i)^{-1/d} \cdot w_\mathbf{1}(V(G_i))
    \leq 3r (n-i)^{1-1/d}.
    \end{align}
    Rename $u_i$ and $v_i$ so that $|\Pp_i(u_i)| \ge |\Pp_i(v_i)|$, where $\Pp_i(w)$ denotes the part of $\Pp_i$ containing~$w$.
    \item\label{itm:resolveStep} For each part $Q\in\Pp_i$ with $L_i(Q)\in\Delta_{G_i}(u_i,v_i)$, resolve $\Pp_i(v_i)$ and~$Q$: positively if $v_i$ and $L_i(Q)$ are adjacent, negatively otherwise. Add the resolved vertex pairs to~$R_i$ to obtain~$R_{i+1}$.

    (In particular, since $u_i \in\Delta_{G_i}(u_i,v_i)$, this step resolves $\PP(u_i)$ and $\PP(v_i)$.)
    \item\label{itm:leaderChoice} Let $\Pp_{i+1}$ be the partition of $V(G)$ formed by merging $\Pp_i(u_i)$ and $\Pp_i(v_i)$ into a single new part, whose leader is~$u_i$. All other parts keep their leaders.
    \item\label{itm:weightUpdate}
        Mirroring the multiplicative weight update strategy of Welzl~\cite{Welzl88, ChazelleWelzl89},
        set for each leader $x$ in the new leader graph $G_{i+1}$,
    \[
    w_{i+1}(x)\coloneqq
    \begin{cases}
    2\max(w_i(u_i),w_i(v_i)) & \text{if } x=u_i,\\
    2w_i(x) & \text{if } x\in \Delta_{G_i}(u_i,v_i)\setminus\{u_i,v_i\},\\
    w_i(x) & \text{otherwise.}
    \end{cases}
    \]
    \end{enumerate}
    \end{tcolorbox}

    \bigskip
    
    We claim that the algorithm indeed computes a construction sequence for the graph $G$. The intuition is that for each unresolved pair $\set{x,y}\notin R_i$ with $x$ and $y$ in different parts, the adjacency between $x$ and $y$ agrees with the adjacency between their leaders $L_i(x)$ and $L_i(y)$. Any time this invariant would be violated by a leader change, the pair is resolved beforehand in Step~\ref{itm:resolveStep}.

    \begin{restatable}[\inAppendix{}]{claim}{buildsGraph}
    \label{clm:buildsGraph}
    The algorithm computes a construction sequence of $G$.
    \end{restatable}

    We next bound the radius-$1$ merge-width of the computed construction sequence. We begin by bounding the number of times the leader of a vertex can change. This  is the only place where we use the choice of which vertex becomes the new leader in step~\ref{itm:leaderChoice} of the algorithm. Essentially, this choice balances the ``tree of how leaders change'', ensuring that it has logarithmic height.

    \begin{restatable}[\inAppendix{}]{claim}{numberOfLeaders}
    \label{clm:numberOfLeaders}
        For each vertex $x \in V(G)$, there are at most $\log_2 n + 1$ different vertices which are ever leaders for the part containing $x$. That is, $|\{L_i(x): i \in \{0,1,\ldots, n-1\}\}|\leq\log_2 n + 1$.
    \end{restatable}

    Finally, we bound for each vertex the number of rounds \(i\) such that its leader can appear in $\Delta_{G_i}(u_i, v_i)$.
    The proof of this claim very closely resembles the multiplicative weight update strategy of Welzl~\cite{Welzl88, ChazelleWelzl89}. This claim is the only place where we use inequality~(\ref{equationMUW}).
    
    \begin{restatable}[\inAppendix{}]{claim}{MWUtotal}
    \label{clm:MWUtotal}
        For each vertex $x \in V(G)$, there are at most $\cal O_{c,d}(n^{1-1/d} + \log{n})$ many integers $i \in \{0,1,\ldots, n-2\}$ so that $L_i(x) \in \Delta_{G_i}(u_i,v_i)$.
    \end{restatable}
    \noindent 
    We will next use \Cref{clm:numberOfLeaders} and \Cref{clm:MWUtotal} to bound the width of the construction sequence. 
    In each round $i$ of the algorithm, multiple  resolve steps are performed, followed by a single merge step.
    During each such step, the set of resolved edges is a subset of $R_{i+1}$ and the width is measured in the partition $\PP_i$ (for the resolve steps) or in $\PP_{i+1}$ (for the merge step).
    The width of a single step can therefore be upper bounded by the width achieved for $\PP_i$ and $R_{i+1}$: the finest partition and the most resolved edges.
    As $\PP_i$ and $\PP_{i+1}$ only differ in one part, up to a difference of $1$, the following claim suffices to bound the width of the entire construction sequence.

    \begin{claim}\label{clm:bound-mw}
      For each $i \in \{0,1,2,\ldots, n-1\}$ and vertex $x\in V(G)$, $x$ has neighbors in at most $\mathcal{O}_{c,d}(n^{1-1/d}\log{n})$ different parts of $\PP_i$ in the graph $(V(G), R_i)$.
    \end{claim}

    \begin{proof}
      Fix $x\in V(G)$.
    Our goal is to bound the number of parts of $\Pp_i$ which contain a vertex \(y\) such that $\set{x,y}\in R_i$ by $\mathcal{O}_{c,d}(n^{1-1/d}\log{n})$.
    Each relevant pair $\set{x,y}$ is resolved at a unique round $j\le i$.
    Disregarding the single part $\PP_i(x)$, we can assume that $x$ and $y$ are in different parts of $\PP_i$.
    Then $x$ and $y$ are also in different parts of the finer partition $\Pp_j$, and exactly one of them is in~$\Pp_j(v_j)$.
    Recall that $v_j$ does \emph{not} become the leader of the merged part.
    We distinguish two cases.

    First, we bound the number of parts of $\Pp_i$ which contain a vertex $y$
    so that $\set{x,y}$ was resolved at time \(j\) with $x\in\Pp_j(v_j)$.
    Then $L_j(x) \neq L_{j+1}(x)$. By \Cref{clm:numberOfLeaders}, there are
    at most $\log_2 n$ such rounds $j$. Moreover, by
    inequality~(\ref{equationSize}), there are at most $3r(n-j)^{1-1/d} \leq 3rn^{1-1/d}$ different
    parts of $\Pp_j$ that $y$ could be contained in. Since $\PP_i$ is coarser than each $\PP_j$, this case contributes at most $3rn^{1-1/d} \log_2{n} \leq \cal O_{c,d}(n^{1-1/d} \log{n})$ parts reachable in $\PP_i$.

    Next, we bound the number of parts of $\Pp_i$ which contain a vertex $y$
    so that $\set{x,y}$ was resolved at time \(j\) with $y\in\Pp_j(v_j)$.
    At each round $j$, there is at most one such part, namely $\Pp_j(v_j)$. Moreover, $L_j(x) \in
    \Delta_{G_j}(u_j, v_j)$ by definition. By \Cref{clm:MWUtotal}, there are
    only $\cal O_{c,d}(n^{1-1/d} + \log{n})$ such rounds~$j$.
\newcommand{\setoff}[2]{\{#1\,:\,#2\}}

    Summing the two cases,
     $|\setoff{\Pp_i(y)}{\set{x,y}\in R_i|}\le \mathcal{O}_{c,d}(n^{1-1/d}\log{n})$, as claimed.
    \end{proof}

    Having bounded the width of the computed construction sequence, it remains to bound the running time. The bottleneck is applying
    \Cref{lem:twoWeightFunctions} in each of the $n$ rounds
    (step~\ref{itm:twoWeightsStep}). Since $G_i$ has at most $n$ vertices and $w_i$
    assigns weight at most $2^n$ to each vertex, each application runs in time
    $\mathcal{O}(n^4)$, giving an overall runtime of $\mathcal{O}(n^5)$.
\end{proof}

\bibliographystyle{plain}
\bibliography{ref}

\appendix

\input{nc}

\section{Radius-1 Merge-Width}
\label{apn:mergeWidth}

In this section we include the proofs that were omitted from \Cref{sec:merge-width} about merge-width.

Before beginning the main proofs, we need to state Haussler's Packing Lemma~\cite{Haussler95}. We state an equivalent formulation in terms of bipartite graphs. This is equivalent to the usual version for set systems: take $X$ to be the ground set and $\{N(v):v \in Y\}$ to be the set system. We refer the reader to~\cite[Lem. 5.14]{discrepancy} for another statement of the lemma.

\begin{lemma}[\cite{Haussler95}]
\label{lem:Haussler}
    For all real numbers $c$ and $d\geq 1$, there exists an integer $t = t(c,d)$ so that for any integer $\delta\geq 1$, if $H=(X,Y,E)$ is a bipartite graph such that\begin{enumerate}
        \item for all distinct vertices $u,v \in Y$, there are at least $\delta$ vertices in $X$ which are adjacent to exactly one of $u$ and $v$, and
        \item for every nonempty $A \subseteq X$, we have $|\set{N(v)\cap A\,:\, v\in Y}| \leq c|A|^d$,
    \end{enumerate}
then $|Y|\leq \max \left(t \cdot (|X|/\delta)^d, 1\right)$.
\end{lemma}

Now we can use \Cref{lem:Haussler} to prove the key lemma about fractional twins.

\existFractionalTwins*
\begin{proof}
    Let $t = t(c+1,d)$ be the integer from \Cref{lem:Haussler}, and write $V$ for $V(G)$.

    Fix a weight function $w\colon V \rightarrow \mathbb{Z}^+$. We define a bipartite graph $H=(X,Y,E)$ with $Y=V$ and $|X|=2w(V)$, as follows. For each vertex $v$ of $G$, we place $2w(v)$ vertices in $X$, called the \emph{copies} of~$v$: $w(v)$ of them have neighborhood $N_G(v)$ in $H$, and the remaining $w(v)$ have neighborhood $\{v\}$ in $H$. We call the former \emph{neighborhood copies} and the latter \emph{self copies} of~$v$.

    We verify condition~(2) of \Cref{lem:Haussler} with constant $c+1$. 
    
    \begin{claim}
      For every nonempty $A \subseteq X$, we have $|\setof{N_H(v)\cap A}{v\in Y}| \leq (c+1)|A|^d$.
    \end{claim}
    \begin{proof}
    Consider a nonempty set $A \subseteq X$, and let $A' \subseteq V$ be the set of vertices of $G$ with at least one copy in $A$. Since each element of $A$ is a copy of exactly one vertex, $|A'| \leq |A|$. For $u \in Y$ with no copy in $A$, the set $N_H(u) \cap A$ consists exactly of the neighborhood copies in $A$ of vertices in $N_G(u)$, which is determined by $N_G(u) \cap A'$. Hence,
    \[
    |\setof{N_H(v) \cap A}{v \in Y}| \leq |A'| + |\setof{N_G(v) \cap A'}{v \in V}| \leq |A| + c|A|^d \leq (c+1)|A|^d.\qedhere
    \]
    \end{proof}

    We apply \Cref{lem:Haussler} contrapositively with $\delta \coloneqq r \cdot n^{-1/d} w(V)$, where $r = r(c,d)$ is chosen below. 
    Using $|X| = 2w(V)$, the bound from \Cref{lem:Haussler} gives
    \begin{align*}
    |Y| \leq \max\!\left(t\cdot \left(\frac{|X|}{\delta}\right)^{\!d}\!, 1\right) \leq \max\!\left(t\cdot \left(\frac{2}{r}\right)^{\!d} n, 1\right).
    \end{align*}

    Choosing $r$ large enough so that $t \cdot (2/r)^d < 1$, the right-hand side is less than $n = |Y|$, a contradiction. So condition~(1) of \Cref{lem:Haussler} must fail: there exist distinct $u, v \in Y$ such that fewer than $\delta$ vertices of $X$ are adjacent to exactly one of $u$ and $v$.
    It remains to show that $w(\Delta_G(u,v)) \leq \delta$.
    The $w(u)$ self copies of $u$, the $w(v)$ self copies of $v$, and for each $x \in \Delta_G(u,v) \setminus \{u,v\}$, the $w(x)$ neighborhood copies of $x$, are all distinct elements of $X$ adjacent to exactly one of $u$ and $v$. Hence $w(\Delta_G(u,v)) = w(u) + w(v) + w(\Delta_G(u,v) \setminus \{u,v\}) < \delta = rn^{-1/d}w(V)$, as desired.
\end{proof}

The next lemma shows how to obtain a single pair of vertices which are fractional twins with respect to two weight functions simultaneously.

\twoWeightFunctions*
\begin{proof}
    Denote $V\coloneqq V(G)$. Without loss of generality, we may assume that $w_1(V) \geq w_2(V)$. Now, for each vertex $x \in V$, set \begin{align*} w(x) = w_1(x)+w_2(x) \cdot \left\lceil {w_1(V)}/{w_2(V)}\right\rceil.\end{align*}Then $w(V)\leq 2w_1(V)+w_2(V)\leq 3 w_1(V)$ where the extra $w_2(V)$ is used to handle parity. Now, let $u$ and $v$ be distinct vertices of $G$ which are $p$-fractional twins with respect to $w$. Then \begin{align*} w_1(\Delta(u,v)) \leq w(\Delta(u,v))\leq p\cdot w(V) \leq 3p\cdot w_1(V),\end{align*} so $u$ and $v$ are $3p$-fractional twins with respect to $w_1$. To prove $3p$-fractional twins with respect to $w_2$, we have by definition 
    \[
      w_2(\Delta(u,v)) \cdot \left\lceil {w_1(V)}/{w_2(V)}\right\rceil = w(\Delta(u,v)) - w_1(\Delta(u,v)).
    \]
    So in particular
    \[
      w_2(\Delta(u,v))\cdot \left( {w_1(V)}/{w_2(V)}\right) \leq w(\Delta(u,v))\leq 3p\cdot w_1(V).
    \]
    Dividing the leftmost and the rightmost side by $\left( {w_1(V)}/{w_2(V)}\right)$, we obtain $w_2(\Delta(u,v)) \leq 3p\cdot w_2(V)$,
    as desired.

    Now, notice that every vertex has weight at most $2^{\mathcal{O}(n)}$ with respect to $w$. So we may find such vertices $u$ and $v$ by iterating through all $\mathcal{O}(n^2)$ pairs of vertices $u$ and $v$, and then computing $w(\Delta(u,v))$ in $\mathcal{O}(n^2)$ time. (Here we represent the weights of size $2^{\mathcal{O}(n)}$ as bitstrings of length $\mathcal{O}(n)$, which we can add and compare in time $\mathcal{O}(n)$.)
\end{proof} 

We conclude this section by proving three claims which are all part of the proof of \Cref{thm:mergeWidth}. The reader may wish to refer back to that proof.

\buildsGraph*
\begin{proof}
    By induction on \(i\), we maintain two invariants at the beginning of round \(i\):
    \begin{itemize}
        \item\label{inv:within} Every pair of vertices within the same part of~$\Pp_i$ is resolved.
        \item\label{inv:leader} For every unresolved pair $\set{x,y} \not\in R_i$ with $x$ and $y$ in different parts of~$\Pp_i$, $x$ is adjacent to $y$ if and only if $L_i(x)$ is adjacent to $L_i(y)$.
    \end{itemize}

    Together, these invariants prove the claim:
    At round~$n-1$, only a single part is left, and by the first item every pair is resolved.
    Moreover, Step~\ref{itm:resolveStep} resolves a pair \(\{x,y\}\) positively if and only if their leaders are adjacent.
    By the second item, this is the case if and only if \(x\) and \(y\) are adjacent.
    Hence, every pair is resolved according to their adjacency in the graph, proving the claim.

    At $i=0$, every vertex is its own leader and every part is a singleton, so
    both items hold trivially. Assume both hold at round~$i$, and we verify them at
    round~$i+1$.

    To prove the first invariant, observe that the only new part of~$\Pp_{i+1}$ is $\Pp_i(u_i)\cup\Pp_i(v_i)$.
    Pairs within~$\Pp_i(u_i)$ or within~$\Pp_i(v_i)$ are already resolved by induction.
    Since $u_i\in\Delta_{G_i}(u_i,v_i)$, Step~\ref{itm:resolveStep} resolves \(\Pp_i(u_i)\) and \(\Pp_i(v_i)\),
    resolving all pairs in $\Pp_i(u_i)\cup\Pp_i(v_i)$ and proving the invariant.
    
    To prove the second invariant, let $\set{x,y}\not\in R_{i+1}$ with $x$ and~$y$ in different parts of~$\Pp_{i+1}$.
    Then $\set{x,y}$ was unresolved also at round~$i$, and $x,y$ are in different parts of~$\Pp_i$ as well (since parts only merge).
    The only leader that changes from round~$i$ to~$i{+}1$ is that of vertices in~$\Pp_i(v_i)$: from~$v_i$ to~$u_i$.
    If neither $x$ nor~$y$ is in~$\Pp_i(v_i)$, then both leaders are unchanged and the invariant carries over from round~$i$.
    So suppose $x\in\Pp_i(v_i)$ (the other case is symmetric).
    Since $x$ and~$y$ are in different parts of~$\Pp_{i+1}$, we have $y\notin\Pp_i(u_i)\cup\Pp_i(v_i)$, so $L_{i+1}(y)=L_i(y)$.
    Since $\set{x,y}$ was not resolved at round~$i$, we have $L_i(y)\notin\Delta_{G_i}(u_i,v_i)$, meaning $L_i(y)$ has the same adjacency to~$u_i$ and~$v_i$.
    Again, the invariant carries over from round~$i$.

\end{proof}

\numberOfLeaders*
\begin{proof}
    Consider a rooted forest on~$V(G)$ that evolves with the algorithm.
    At round~$0$, there are $n$ single-vertex trees, one per part.
    At round~$i$, when parts $\Pp_i(u_i)$ and $\Pp_i(v_i)$ merge, we make~$v_i$ a child of~$u_i$, joining the two trees with~$u_i$ as root.
    After all $n{-}1$ rounds, this yields a single rooted tree~$T$ on~$V(G)$.
    The distinct leaders of a vertex~$x$ across all rounds are exactly its ancestors in~$T$ (including~$x$ itself), so the claim is equivalent to~$T$ having height at most~$\log_2 n$.

    We show by induction on~$i$ that for each part $P \in \Pp_i$, the tree corresponding~$P$ has height at most~$\log_2 |P|$.
    At round~$0$ this holds since each tree has height~$0$.
    When $\Pp_i(u_i)$ and $\Pp_i(v_i)$ merge, the height of the new tree is $\max(h_u,\, 1 + h_v)$, where $h_u$ and~$h_v$ are the heights of the two subtrees.
    By induction, $h_u \leq \log_2 |\Pp_i(u_i)| \leq \log_2 |\Pp_i(u_i) \cup \Pp_i(v_i)|$.
    Since $|\Pp_i(u_i)| \geq |\Pp_i(v_i)|$, we have $1 + h_v \leq 1 + \log_2 |\Pp_i(v_i)| = \log_2(2|\Pp_i(v_i)|) \leq \log_2(|\Pp_i(u_i)| + |\Pp_i(v_i)|)$.
    Both are at most $\log_2 |\Pp_i(u_i) \cup \Pp_i(v_i)|$, completing the induction.
\end{proof}

\MWUtotal*
\begin{proof} 
    Denote the total weight as $t_i\coloneqq w_i(V(G_i))$ for each $i$. Then $t_0=n$.
    We first show that for every $i\in\{0,1,\ldots,n-2\}$,
    \[
    t_{i+1}\leq t_i+w_i(\Delta_{G_i}(u_i,v_i)).
    \]
    Indeed, by Step~\ref{itm:weightUpdate} of the algorithm, every leader in $\Delta_{G_i}(u_i,v_i)\setminus\{u_i,v_i\}$ doubles its weight, contributing an increase of exactly
    $w_i(\Delta_{G_i}(u_i,v_i)\setminus\{u_i,v_i\})$.
    Moreover, by Step~\ref{itm:weightUpdate}, $v_i$ disappears and $u_i$ receives weight $2\max(w_i(u_i),w_i(v_i))$, so the net change on $\{u_i,v_i\}$ is
    \[
    2\max(w_i(u_i),w_i(v_i))-w_i(u_i)-w_i(v_i)
    =|w_i(u_i)-w_i(v_i)|
    \leq w_i(u_i)+w_i(v_i).
    \]
    Hence the displayed inequality holds. Using~(\ref{equationMUW}), we obtain
    \[
    t_{i+1}\leq t_i\left(1+3r(n-i)^{-1/d}\right)
    \qquad\text{for every }i\in\{0,1,\ldots,n-2\}.
    \]
    Therefore, using $t_0=n$,
    \[
    t_{n-1}\leq n\prod_{i=0}^{n-2}\left(1+3r(n-i)^{-1/d}\right)
    =n\prod_{i=2}^{n}\left(1+3ri^{-1/d}\right).
    \]

    Let
    \[
    m\coloneqq\left|\bigl\{i\in\{0,1,\ldots,n-2\}: L_i(x)\in \Delta_{G_i}(u_i,v_i)\bigr\}\right|.
    \]
    Every time $L_i(x)\in \Delta_{G_i}(u_i,v_i)$, the weight of the current leader of $x$ at least doubles from round $i$ to round $i+1$.
    Iterating over the $m$ relevant rounds, we obtain \(w_{n-1}(L_{n-1}(x))\geq 2^m\), and therefore (using \(1+y\leq e^y\)),
    \[
    2^m\leq w_{n-1}(L_{n-1}(x))\leq t_{n-1}
    \leq n\prod_{i=2}^{n}\left(1+3ri^{-1/d}\right)
    \leq n\exp\left(3r\sum_{i=2}^{n} i^{-1/d}\right).
    \]
    Since $i^{-1/d}$ is decreasing,
    \[
    \sum_{i=2}^{n} i^{-1/d}\leq \int_1^n i^{-1/d}\,di.
    \]
    If $d=1$, this integral is $\log n$, so $m\le \mathcal{O}_{c}(\log n)$.
    If $d>1$, this integral is $\mathcal{O}_d(n^{1-1/d})$, so $m\le \mathcal{O}_{c,d}(n^{1-1/d})$.
    In either case
    \[
    m\leq \cal O_{c,d} (n^{1-1/d} + \log n).\qedhere
    \]
\end{proof}



\end{document}

%% file: nc.tex
\section{Neighborhood Complexity}
\label{apn:nc}

In this appendix, we give the full proof of \Cref{thm:nei-comp-NIP}.
We begin by introducing the relevant notions.
\subsection{Preliminaries}
\label{apn:prelims}

\paragraph{Asymptotic notation.}
Throughout this section, $\subpoly(n)$ denotes an unspecified function $f\from\N\to\mathbb R_{\ge 0}$ such that for every fixed $\eps>0$ we have $f(n)\le n^\eps$ for large enough $n$. We write $\subpoly_k(n)$ when $f$ may additionally depend on a fixed parameter $k$, and similarly we write $\subpoly_{\CC}(n)$ when $f$ may depend on a fixed class of graphs $\CC$. (So this notation is the same as writing $n^{o(1)}$, $n^{o_k(1)}$, or $n^{o_{\CC}(1)}$. However, we introduce the new notation as to make the variables and parameters explicit.) Similarly, $\polylog(n)$ denotes an unspecified function bounded by $p(\log n)$ for some polynomial $p$, and $\polylog_k(n)$ allows dependence of the polynomial on a fixed parameter~$k$.

\paragraph{Graphs.}
A graph $G$ consists of a set $V$ of vertices and a set $E\subset {V\choose 2}$ of edges, which are denoted $V(G)$ and $E(G)$, respectively.
For a graph $G$ with disjoint subsets $A,B \subseteq V(G)$,  let $G[A,B]$ denote the bipartite graph $(A,B,E')$ that is \emph{semi-induced} by $A$ and $B$, where $E'$ contains those edges of $G$ with one endpoint in $A$ and one endpoint in $B$.

\paragraph{Structures.}
A graph is viewed as a relational structure, equipped with a binary relation symbol $E$ denoting adjacency.
A \emph{unary expansion} of a relational structure $A$ is obtained from $A$ by 
adding several unary relations (also called unary predicates) to~$A$.

\subsection{A simple lemma on VC-dimension}
\lemVCind*
\begin{proof}[Proof of \Cref{lem:VC-induction}]
It is enough to prove the claim for the family $\cal F_+$ of all $v$-positive sets in $\cal F$; the argument for the family of $v$-negative sets is symmetric.
If $\cal F_+=\emptyset$, then $\VCdim(\cal F_+)=-\infty$ and there is nothing to prove.

Suppose $\cal F_+$ shatters a set $X\subseteq V$.
Since every set in $\cal F_+$ contains $v$, we must have $v\notin X$.
For each $Y\subseteq X$, choose $F_Y\in\cal F_+$ with $F_Y\cap X=Y$.
As $F_Y$ is $v$-positive, there is a set $F'_Y\in\cal F$ such that $F'_Y\triangle F_Y=\set{v}$.
Then $F'_Y\cap X=Y$, while exactly one of $F_Y$ and $F'_Y$ contains $v$.
Hence
\[
\{F_Y\cap (X\cup\set{v}),\ F'_Y\cap (X\cup\set{v})\}=\set{Y\cup\set{v},\ Y}.
\]
Since this holds for every $Y\subseteq X$, the family $\cal F$ shatters $X\cup\set{v}$.
Therefore,
\(
\VCdim(\cal F)\ge |X|+1,
\)
so $\VCdim(\cal F_+)\le \VCdim(\cal F)-1$.
This proves that $\cal F_+$ has VC-dimension strictly smaller than~$\cal F$.
\end{proof}

\subsection{Exhibiting many Hamming edges}
We prove the following generalization of \Cref{lem:star0}.

\begin{restatable}
  {lemma}{lemstar}\label{lem:star}
  Let $G=(A,B,E)$ be a bipartite graph with $|A|=n>1$, and let $\cal P$ be a partition of $B$.
  Suppose
  \[
    m:=\sum_{P\in \cal P}|\inducedSS{P}{A}|> 2|\cal P|.
  \]
  Then there is a set $X\subseteq A$ such that
  \[
    \sum_{P\in \cal P}|E(\hamming(\inducedSS{P}{X}))|
    \ge \frac{m}{2\ln n+2}.
  \]
\end{restatable}
\begin{proof}
  Let $H(n):=1+1/2+\ldots+1/n$ be the $n$th harmonic number.
  Initially, set $X:=A$.
  As long as there is a vertex $v\in X$ such that 
  \[
    \sum_{P\in\cal P}\left(|\inducedSS{P}{X}|-|\inducedSS{P}{(X-\set v)}|\right)
    \le \frac{m}{2|X|\cdot H(n)},
  \]
  remove $v$ from $X$ and repeat. Otherwise, if there is no such $v$, terminate.

  \begin{claim}
    At the end of the process, the set $X$ is nonempty.
  \end{claim}
  \begin{proof}
    At any point in the process, consider the sum
    \[
      \mu(X):=\sum_{P\in \cal P}|\inducedSS{P}{X}|.
    \]
    Initially, when $X=A$, this sum is equal to $m$ by definition.
    In each step of the process, when a vertex $v$ is removed from $X$, the sum decreases by
    \[
      \sum_{P\in \cal P}\Big(|\inducedSS{P}{X}|-|\inducedSS{P}{(X-\set v)}|\Big)
      \le \frac{m}{2H(n)}\cdot \frac{1}{|X|}.
    \]
    Hence, if we remove all $n$ elements from $X$, the sum decreases by at most
    \[
      \left(\frac{m}{2H(n)}\cdot \frac{1}{n}\right) + \left(\frac{m}{2H(n)}\cdot \frac{1}{n-1}\right) + \ldots + \left(\frac{m}{2H(n)}\cdot \frac{1}{1}\right) = \frac{m}{2H(n)}\cdot \sum_{i=1}^n \frac 1 i
      = \frac{m}{2}.
    \]
    This means $\mu(\emptyset)\geq \mu(A) - m/2 = m/2$.
    On the other hand we have
    \[
      \mu(\emptyset)
      =\sum_{P\in\cal P}|\inducedSS{P}{\emptyset}|
      =\sum_{P\in\cal P}1
      =|\cal P|
      <\frac{m}{2};
    \]
    a contradiction which proves the claim.
  \end{proof}

  For every part $P\in\cal P$ and every vertex $v\in X$, the quantity
  \[
    |\inducedSS{P}{X}|-|\inducedSS{P}{(X-\set v)}|
  \]
  is exactly the number of edges of $\hamming(\inducedSS{P}{X})$ whose endpoints
  differ on~$v$. Therefore every edge of every graph
  $\hamming(\inducedSS{P}{X})$ is counted exactly once when we sum over
  $v\in X$, and hence
  \begin{multline*}
    \sum_{P\in\cal P}|E(\hamming(\inducedSS{P}{X}))|
    = \sum_{v\in X}
    \sum_{P\in\cal P}\left(|\inducedSS{P}{X}|-|\inducedSS{P}{(X-\set v)}|\right)\\
    > \sum_{v\in X}\frac{m}{2|X|\cdot H(n)}
    = |X|\cdot \frac{m}{2|X|\cdot H(n)}
    = \frac{m}{2H(n)}
    \ge \frac{m}{2\ln(n)+2}.
  \end{multline*}
\end{proof}

\subsection{Proof of \Cref{thm:nei-comp-NIP}}

 \Cref{thm:nei-comp-NIP} will follow directly from the next lemma.
\begin{restatable}{lemma}{lemnc}\label{lem:nc}
  Let $\BB$ be a monadically dependent class of bipartite graphs $G=(A,B,E)$
  with no twins in $B$.
  Then \begin{equation}\label{eq:bip-version}
  |B|\le |A|\cdot \subpoly_{\BB}(|A|).
\end{equation}
\end{restatable}
We first argue that \Cref{lem:nc} implies 
\cref{thm:nei-comp-NIP}.

\begin{proof}[Proof of \Cref{thm:nei-comp-NIP}]
Let $\BB$ be the class of all bipartite graphs $G[A,B]$
such that $G\in \CC$ and $A,B\subseteq V(G)$ are disjoint,
and no two vertices in $B$ have equal neighborhoods in $A$.
As $\CC$ transduces $\BB$, the class $\BB$ is monadically dependent.
For a graph $G\in \CC$ and set $A\subseteq V(G)$, let $B\subset V(G)-A$ 
be maximal such that no two vertices in $B$ have equal neighborhoods in $A$.
Then
\[
  \Big|\setof{N_G(v)\cap A}{v\in V(G)}\Big|\le  \Big|\setof{N_G(v)\cap A}{v\in A}\Big|+\Big|\setof{N_G(v)\cap A}{v\in B}\Big|   
  \le |A|+|B|,
\]
 Applying \Cref{lem:nc} to $G[A,B]\in\BB$ yields the theorem.
\end{proof}

We now proceed to the proof of \Cref{lem:nc}.
The following notion encapsulates the invariant used in the inductive proof.
\begin{definition}
  Let $G=(A,B,E)$ be a bipartite graph and $k\in\N$.
  A \emph{$k$-sparsification} consists of:
  \begin{itemize}
    \item sets $A_0\subseteq A$ and $B_0\subseteq B$, and
    \item functions $f_1,\ldots,f_k\from B_0\to A$,
  \end{itemize}
  such that
  the mapping $b\mapsto (N(b)\cap A_0,f_1(b),\ldots,f_k(b))$ is an injection from $B_0$ to $2^{A_0} \times A^k$.
For such a $k$-sparsification we define:
\begin{itemize}
  \item The \emph{associated partition} $\cal P$ as the partition of $B_0$ such that two vertices $b,b'$ are in the same part of $\cal P$ if and only if 
  $f_j(b)=f_j(b')$ for all $j=1,\ldots,k$. Note that vertices in the same part of $\cal P$ have pairwise distinct neighborhoods in $A_0$;
  \item The \emph{size} as $s=|B_0|$;
  \item The \emph{dimension} as $\max_{P\in\cal P} \VCdim(\inducedSS{P}{A_0})$; and
  \item The \emph{complexity} as the least $c\ge0$ such that there are first-order formulas $\phi_1(x,y),\ldots,\phi_k(x,y)$ of quantifier rank ${\le}c$, each using ${\le}c$ unary predicates, and there is an expansion $G'$ of $G$ with $c$ unary predicates,
  such that $\phi_j$ defines $f_j$ in $G'$ for $j=1,\ldots,k$, so that 
  \[
    f_j(b)=a\iff G'\models \phi_j(b,a)
    \quad
    \text{for all $b\in B_0$ and $a\in A$.}  
  \]
\end{itemize}
We say that a sparsification is \emph{terminal} if its size $s$ and
associated partition $\cal P$ satisfy $s\le 2|\cal P|$, and
\emph{nonterminal} otherwise.
\end{definition}

Note that a bipartite graph $G=(A,B,E)$ with no twins in $B$ has a trivial $0$-sparsification with $A_0=A$ and $B_0=B$. Its associated partition is $\cal P=\set{B}$, size $s=|B|$, and complexity is $0$.
We prove the following, for bipartite graphs $G$ from a monadically dependent class $\BB$:
\begin{enumerate}
  \item a nonterminal $k$-sparsification can be improved to a $(k+1)$-sparsification with strictly smaller dimension, by losing only a $\polylog(|A|)$ factor in the size, and increasing the complexity only by a constant (\Cref{lem:step});
  \item repeating this argument, we reach a terminal sparsification after at most $d$ steps, where $d$ upper bounds the VC-dimension of every $G$ in $\BB$ (\Cref{lem:iterate});
  \item  a terminal $k$-sparsification of $G\in \BB$ of bounded complexity has size $s\le |A|\cdot \subpoly_{\BB,k}(|A|)$ (\Cref{lem:terminal}).
\end{enumerate}
Combining these three points yields a terminal $k$-sparsification with $k\le d$ and size $s$ satisfying
$$\frac{|B|}{\polylog_d(|A|)}\le s\le |A|\cdot\subpoly_{\BB,d}(|A|).$$
This proves the inequality \eqref{eq:bip-version} in \Cref{lem:nc}. We now prove the necessary lemmas.

\begin{restatable}
  {lemma}{lemtransduce}\label{lem:transduce}
  For every $k,c\in\N$ there is a transduction $T_{k,c}$ with the following property.
  Let $G=(A,B,E)$ be a bipartite graph with a $k$-sparsification 
   of complexity $c$, consisting of sets 
   $A_0\subseteq A,B_0\subseteq B$ and functions $f_1,\ldots,f_k\from B_0\to A$.
   Consider the bipartite graph $H=(A,B_0,E')$
with $N_H(b)=\setof{f_j(b)}{j\in[k]}$ for $b\in B_0$.
Then $H\in T_{k,c}(G)$.
\end{restatable}
\begin{proof}
Fix $k,c\in\N$.
Up to logical equivalence, there are only finitely many binary formulas of
quantifier rank at most $c$ over the signature consisting of the adjacency
relation and $c$ unary predicates.
Hence there are only finitely many $k$-tuples of such formulas; list them as
\[
  \bar\phi^1,\ldots,\bar\phi^m,
  \qquad
  \bar\phi^t=(\phi^t_1,\ldots,\phi^t_k).
\]

Let $T_{k,c}$ be the transduction that guesses:
\begin{itemize}
  \item $c$ unary predicates for the witness expansion from the definition of complexity,
  \item two unary predicates $L,R$ marking the sets $A$ and $B_0$, and
  \item unary predicates $S_1,\ldots,S_m$ used as global flags selecting one
  of the tuples $\bar\phi^1,\ldots,\bar\phi^m$.
\end{itemize}
It then applies the formula
\[
  \psi(x,y)\ :=\ L(x)\land R(y)\land
  \bigvee_{t=1}^m\left(\exists z\,S_t(z)\land \bigvee_{j=1}^k\phi_j^t(y,x)\right),
\]
where the inner disjunction is interpreted as false when $k=0$,
and finally takes the induced subgraph on $L\cup R$.

Now suppose $G$ comes with a $k$-sparsification of complexity $c$ as in the
statement.
By definition of complexity, after choosing a suitable expansion of $G$ by
$c$ unary predicates, there is some index $t\in[m]$ such that
$\phi_1^t,\ldots,\phi_k^t$ define the functions $f_1,\ldots,f_k$ on~$B_0$.
Interpret $L$ as $A$, $R$ as $B_0$, interpret $S_t$ as $V(G)$, and interpret
all other flags $S_{t'}$ as empty.
For this choice of unary predicates, the graph produced by $\psi$ has exactly
the edges $\{a,b\}$ with $a\in A$, $b\in B_0$, and $a=f_j(b)$ for some $j\in[k]$.
Therefore, the induced subgraph on $A\cup B_0$ is precisely the graph
$H=(A,B_0,E')$ with
\[
  N_H(b)=\setof{f_j(b)}{j\in[k]}
  \qquad\text{for every }b\in B_0.
\]
Hence, $H\in T_{k,c}(G)$.
\end{proof}

\begin{lemma}\label{lem:terminal}
  Let $\BB$ be a monadically dependent class of bipartite graphs and let $k,c\in\N$.
  Let $G=(A,B,E)\in\BB$ have a terminal $k$-sparsification of complexity at most $c$ and size $s$.
  Then
  \[
    s\le |A|\cdot \subpoly_{\BB,c,k}(|A|).
  \]
\end{lemma}
\begin{proof}
Let $A_0\subseteq A$, $B_0\subseteq B$, and $f_1,\ldots,f_k\from B_0\to A$
witness the given $k$-sparsification,
and let $\cal P$ be the associated partition. We prove that $|\cal P|\le |A|\cdot \subpoly_{\BB,c,k}(|A|)$, which will imply the claim as $s\le 2|\cal P|$ by terminality.
Choose a set $B^*\subseteq B_0$ containing exactly one vertex from each part of the associated partition
$\cal P$.

By \Cref{lem:transduce}, there is a transduction $T_{k,c}$ such that the
bipartite graph $H_0$ with parts $A$ and $B_0$ and
edges $\setof{\set{b,f_j(b)}}{b\in B_0,j\in[k]}$ 
belongs to $T_{k,c}(G)$.
Since transductions allow taking induced subgraphs, the graph
$H\coloneqq H_0[A\cup B^*]$ also belongs to $T_{k,c}(G)$.

Each vertex of $B^*$ has degree at most $k$ in $H$, so $H$ is
$K_{k+1,k+1}$-free.
Moreover, distinct vertices of $B^*$ lie in different parts of $\cal P$, hence
their tuples
\[
  (f_1(b),\ldots,f_k(b))
\]
are pairwise distinct.
Therefore, for every set $S\subseteq A$, there are at most $|S|^k$ vertices
$b\in B^*$ with $N_H(b)=S$, because each coordinate of the above tuple must
then belong to $S$.
Since every such $S$ has size at most $k$, there are in fact at most
$k^k$ such vertices (with $0^0=1$). Consequently, 
\[
  |\cal P|=|B^*|
  \le k^k\cdot \Big|\setof{N_H(b)}{b\in B^*}\Big|.
\]

Since $T_{k,c}(\BB)$ is again monadically dependent and $H\in T_{k,c}(\BB)$,
applying \Cref{thm:weaklySparseCase} with $t=k+1$ yields
\[
  \Big|\setof{N_H(b)}{b\in B^*}\Big|
  =\Big|\setof{N_H(b)\cap A}{b\in B^*}\Big|
  \le |A|\cdot \subpoly_{\BB,c,k}(|A|).
\]
The two inequalities,
together with $s\le 2|\cal P|$,
prove the
statement.
\end{proof}

The next lemma is the key ingredient of the proof of \Cref{thm:nei-comp-NIP}.
\begin{lemma}\label{lem:step}
  Suppose $G=(A,B,E)$ has a nonterminal $k$-sparsification of size $s$,
  dimension $d$, and complexity $c$.
  Then there is a $(k+1)$-sparsification of $G$ with:
  \begin{itemize}
    \item size at least $\Omega(\frac{s}{d\cdot \log^2(|A|)})$,
    \item dimension at most $d-1$,
    \item complexity at most $c+5$.
  \end{itemize}
\end{lemma}
\begin{proof}
Let $A_0\subseteq A$, $B_0\subseteq B$, and $f_1,\ldots,f_k\from B_0\to A$
witness the given $k$-sparsification.
Let $\cal P$ be its associated partition.
Since the sparsification is nonterminal, we have $s>2|\cal P|$.
For every part $P\in\cal P$, the functions $f_1,\ldots,f_k$ are constant on $P$,
so the injectivity condition in the definition of a sparsification implies that
the map $b\mapsto N_G(b)\cap A_0$ is injective on $P$.
Therefore,
\[
  |\inducedSS{P}{A_0}|=|P|
  \qquad\text{for every }P\in\cal P,
\]
and hence
\[
  \sum_{P\in\cal P}|\inducedSS{P}{A_0}|=\sum_{P\in\cal P}|P|=|B_0|=s.
\]
In particular, $d\ge 1$, since otherwise every set system $\inducedSS{P}{A_0}$
would have VC-dimension $0$ and hence size $1$, contradicting
$s>2|\cal P|$.
In particular, $|A_0|>1$, since otherwise every set system $\inducedSS{P}{A_0}$
would have size at most $2$, contradicting $s>2|\cal P|$.

Apply \cref{lem:star} to the bipartite graph $G[A_0,B_0]$ and the partition
$\cal P$ of $B_0$. Since
\[
  \sum_{P\in\cal P}|\inducedSS{P}{A_0}|=s>2|\cal P|,
\]
we obtain a set $A_1\subseteq A_0$ such that
\begin{equation}\label{eq:step-many-hamming-edges}
  \sum_{P\in\cal P}|E(\hamming(\inducedSS{P}{A_1}))|
  \ge \frac{s}{2\ln |A_0|+2}.
\end{equation}

Choose a set $B^*\subseteq B_0$ so that for each part $P\in\cal P$,
\[
  \inducedSS{(P\cap B^*)}{A_1}=\inducedSS{P}{A_1}
\]
and no two vertices in $P\cap B^*$ have equal neighborhoods in $A_1$.
Thus, the map $b\mapsto N_G(b)\cap A_1$ is a bijection from $P\cap B^*$ to
$\inducedSS{P}{A_1}$.

For each $P\in\cal P$, let $H_P$ be the Hamming graph of
$\inducedSS{P}{A_1}$, with vertex set identified with $P\cap B^*$ (identified using the previously constructed bijection between the two).
Let $H^*$ be the disjoint union of the graphs $H_P$, for $P\in\cal P$.
Then $V(H^*)=B^*$ and, by \eqref{eq:step-many-hamming-edges},
\[
  |E(H^*)|
  =\sum_{P\in\cal P}|E(H_P)|
  \ge \frac{s}{2\ln |A_0|+2}.
\]

For each $P\in\cal P$, let $M_P^+$ and $M_P^-$ denote the positive and
negative merge graphs of $\inducedSS{P}{A_1}$.
Let $M^+$ and $M^-$ be the disjoint unions of the graphs
$M_P^+$ and $M_P^-$, respectively, identified along the common part $A_1$.
As $A_1\subset A_0$,
 we have $\VCdim(\inducedSS{P}{A_1})\le \VCdim(\inducedSS{P}{A_0})\le d$ for every $P\in\cal P$.
Then \Cref{cor:hamming} implies that $H_P$ has at least $|E(H_P)|/d$
non-isolated vertices.
Summing over $P\in\cal P$, we infer that $H^*$ has at least
\[
  \frac{|E(H^*)|}{d}
  \ge \frac{s}{d(2\ln |A_0|+2)}
\]
non-isolated vertices.
Every such vertex is non-isolated in at least one
of $M^+$ and $M^-$, so for some $\sigma\in\set{+,-}$ the graph $M^\sigma$
has at least
\begin{equation}\label{eq:step-many-nonis}
  \frac{s}{2d(2\ln |A_0|+2)}
\end{equation}
non-isolated vertices in $B^*$.

We now choose a large subset on which each vertex has a unique
$M^\sigma$-neighbor.
We apply \cref{lem:dreier-sampling} to the bipartite graph obtained
from $M^\sigma$ by removing the isolated vertices. 
We obtain sets $A_1'\subseteq A_1$ and $B_1\subseteq B^*$ such that every
vertex of $B_1$ has exactly one neighbor in $A_1'$ in $M^\sigma$, and
\[
  |B_1|
  \ge \Omega\left(\frac{s}{d\cdot \log^2|A|}\right),
\]
because $|A_1|\le |A_0|\le |A|$.

Define $f_{k+1}\from B_1\to A_1'$ by mapping each $b\in B_1$ to its unique
neighbor in $A_1'$ in the graph $M^\sigma$.
Let $\cal P_1$ be the family of all nonempty sets of the form
$f_{k+1}^{-1}(v)\cap P$, where $P\in\cal P$ and $v\in A_1'$.

We claim that $A_1$, $B_1$, and the functions
$f_1,\ldots,f_k,f_{k+1}$ form a $(k+1)$-sparsification of $G$.
Indeed, if $b,b'\in B_1$ lie in different parts of $\cal P$, then
$(f_1(b),\ldots,f_k(b))\neq (f_1(b'),\ldots,f_k(b'))$.
If they lie in the same part $P\in\cal P$, then 
$N_G(b)\cap A_1\neq N_G(b')\cap A_1$ whenever $b\neq b'$.
This is because we have chosen $B^* \supseteq B_1$ such that all vertices in $P \cap B^*$ have pairwise different neighborhoods on $A_1$ in $G$.
Hence, the map
\[
  b\mapsto (N_G(b)\cap A_1,f_1(b),\ldots,f_k(b),f_{k+1}(b))
\]
is injective on $B_1$.
Moreover, by construction, $\cal P_1$ is exactly the associated partition of
this new sparsification.

We now verify the claimed properties.
First, let $P'\in\cal P_1$.
Then $P'=f_{k+1}^{-1}(v)\cap P$ for some $P\in\cal P$ and $v\in A_1'$.
The set system $\inducedSS{P'}{A_1}$ is a subfamily of the neighborhood of $v$
in the $\sigma$-merge graph of $\inducedSS{P}{A_1}$, so
\Cref{lem:blue-vc} yields
\[
  \VCdim(\inducedSS{P'}{A_1})\le \VCdim(\inducedSS{P}{A_1}) - 1 \le \VCdim(\inducedSS{P}{A_0}) - 1 \le d-1,
\]
where the second inequality holds as $A_1\subset A_0$.
Thus, the new sparsification has dimension at most $d-1$.
Second, the size of the new sparsification is
\[
  s'\coloneqq |B_1|,
\]
and by the choice of $B_1$ we have
\[
  s' \ge \Omega\left(\frac{s}{d\cdot \log^2|A|}\right).
\]

Finally, let $G'$ be an expansion of $G$ witnessing that the original
$k$-sparsification has complexity $c$,
and let $\phi_1,\ldots,\phi_k$ be formulas defining
$f_1,\ldots,f_k$ in $G'$.
Expand $G'$ further by five unary predicates marking the sets
$A_1$, $A_1'$, $B_1$, $B^*$, and a set $S$ which is all of $V(G)$ if
$\sigma=+$ and empty if $\sigma=-$.
Using the formulas $\phi_1,\ldots,\phi_k$, the relation of belonging to the
same part of $\cal P$ is first-order definable with quantifier rank at most
$c+1$.
Consequently, the edge relation of $M^\sigma$ on $B^*\times A_1$ is also
first-order definable with quantifier rank at most $c+2$:
for $b\in B^*$ and $a\in A_1$, we ask whether there exists
$b'\in B^*$ in the same part of $\cal P$ such that
$N_G(b)\cap (A_1-\set{a})=N_G(b')\cap (A_1-\set{a})$ and
$b,b'$ differ on $a$ in the direction prescribed by $\sigma$.
Restricting this relation to $B_1\times A_1'$ gives a formula defining the
function $f_{k+1}$, of quantifier rank at most $c+2$.
Thus the new sparsification has complexity at most $c+5$.
(We only used quantifier rank $c + 2$, but $c+5$ many unary predicates.)
\end{proof}

\begin{lemma}\label{lem:iterate}
  Let $G=(A,B,E)$ be a bipartite graph with no twins in $B$, with
  $|A|>1$, and with $B\neq\emptyset$. 
  Suppose $\VCdim(\inducedSS{B}{A})\le d$ for some $d$.
  Then for some $i\le d$, the graph $G$ has a terminal $i$-sparsification 
  of complexity at most $5i$ 
  and size $s$ such that
  \[
    s\ge \frac{|B|}{\polylog_d(|A|)}.
  \]
\end{lemma}
\begin{proof}
  Set $n\coloneqq|A|$ and $m\coloneqq|B|$.
Let $\mathbf S_0$ be the trivial $0$-sparsification of $G$, given by
$A_0\coloneqq A$ and $B_0\coloneqq B$.
Since there are no twins in $B$, this is indeed a sparsification, with
size $s_0=m$, dimension at most $d$, and complexity $c_0=0$.

For $j=0,1,2,\ldots$, construct $\mathbf S_j$ recursively as follows.
Write $s_j$, $d_j$, and $c_j$ for the size, dimension, and
complexity of $\mathbf S_j$, and write $\cal P_j$ for its associated
partition.

If $\mathbf S_j$ is terminal, or if $j=d$, stop and set $i\coloneqq j$.
Otherwise, \Cref{lem:step} yields a $(j+1)$-sparsification
$\mathbf S_{j+1}$ of size $s_{j+1}$ and complexity $c_{j+1}$ such that
\[
  s_{j+1}
  \ge \Omega\left(\frac{s_j}{d_j\cdot \log^2 n}\right)
  \ge \Omega\left(\frac{s_j}{d\cdot \log^2 n}\right),
\]
while $c_{j+1}\le c_j+5$ and $d_{j+1}\le d_j-1$.

Let $0 \leq i \leq d$ be the index for which the above construction stops.
Repeated application of the above bounds
gives
\[
  s_i\ge \frac{m}{\polylog_d(n)}
  \qquad\text{and}\qquad
  c_i\le 5i\qquad\text{and}\qquad d_i\le d-i.
\]
We argue that $\mathbf S_i$ is terminal,
which will yield the conclusion for $s\coloneqq s_i$.

If $i<d$, then $\mathbf S_i$ is terminal by construction.
Suppose now that $i=d$.
Then $\mathbf S_d$ has dimension~$0$.
Let $P\in\cal P_d$.
By definition of the associated partition, the functions
$f^d_1,\ldots,f^d_d$ are constant on $P$.
Since $\mathbf S_d$ is a sparsification, the map
$b\mapsto N_G(b)\cap A_d$ is injective on~$P$. 
Hence, the vertices of $P$ have pairwise different neighborhoods on $A_d$, and we have $|\inducedSS{P}{A_d}| = |P|$.
Furthermore, as $d_i=0$, the set system $\inducedSS{P}{A_d}$ has VC-dimension $0$,
and hence size $1$.
Thus, also $|P|=1$ for every $P\in\cal P_d$, and therefore
\[
  s_d=|B_d|=|\cal P_d|\le 2|\cal P_d|.
\]
So $\mathbf S_d$ is terminal as well.
\end{proof}

We now prove \cref{lem:nc}, which we restate here for convenience.

\lemnc*

\begin{proof}
Fix now a graph $G=(A,B,E)\in\BB$,
and let $n\coloneqq |A|$ and $m\coloneqq |B|$.
If $|A|\le 1$ or $|B|=0$ then the statement is trivial, so suppose $|A|\ge 2$ and $|B|\ge 1$.
As $\BB$ is monadically dependent, the set system $\inducedSS{B}{A}$ has
VC-dimension bounded by a constant $d$ depending only on~$\BB$.
By \Cref{lem:iterate}, for some $i\le d$, the graph $G$ has a
terminal $i$-sparsification of size $s$ and complexity at most $5i$ such
that
\[
  s\ge \frac{m}{\polylog_d(n)}.
\]
By \Cref{lem:terminal}, as $i\le d$, we have
\[
  s
  \le n\cdot \subpoly_{\BB,d}(n).
\]
Combining the two inequalities, we obtain
\[
  m
  \le n\cdot \polylog_d(n)\cdot \subpoly_{\BB,d}(n)\le n\cdot \subpoly_{\BB}(n),
\]
 which proves  \eqref{eq:bip-version}.
\end{proof}